\documentclass{article}
\usepackage{ulem}

\newcommand\redout{\bgroup\markoverwith
{\textcolor{red}{\rule[.5ex]{4pt}{1.4pt}}}\ULon}
\usepackage{xcolor}

\usepackage{amsmath}
\usepackage{amssymb}
\usepackage{amsthm}
\usepackage{graphicx}

\newtheorem{assumption}{Assumption}
\newtheorem{theorem}{Theorem}

\newtheorem{lemma}{Lemma}
\newtheorem{proposition}{Proposition}

\usepackage{ulem}

\usepackage{color}

\def\BFPp{\text{BFP}^+}
\def\BFPn{\text{BFP}^-}
\def\Zp{{Z^+}}
\def\Zn{{Z^-}}
\def\sg{{s_g}}
\def\sz{{s_z}}
\def\Var{\text{Var}}
\def\Cov{\text{Cov}}
\def\rhogz{{\rho_{G,Z}}}
\def\EN{h}
\def\ENtwo{v}
\def\ENplus{{h_+(p)}}
\def\ENplusprime{h_+'(0)}
\def\ENplusprimeprime{h_+''(0)}

\def\ENplusprimefull{\frac{d}{dp}\ENplus|_{p=0}}
\def\ENplusprimeprimefull{\frac{d^2}{dp^2}\ENplus|_{p=0}}
\def\alphap{\alpha^+}
\def\betamax{\beta_\text{max}}
\def\hmin{h_\text{min}}

\def\ch#1{\frac{\alpha'(#1)(#1-1)}{\alpha(#1)}}

\def\R{ {\mathbb{R} }}

\def\Fig#1{Fig.~\ref{#1}}
\def\FigR#1 #2 {Fig.~\ref{#1}-\ref{#2}}

\def\eq#1{eq.~\eqref{#1}}

\begin{document}

\title{Inferring average generation via division-linked labeling}

\author{
Tom S. Weber\thanks{Hamilton Institute, Maynooth University, Ireland},
Le\"ila Peri\'e\thanks{Division of Immunology, Netherlands Cancer
Institute \& Theoretical Biology and Bioinformatics, Utrecht
University, the Netherlands. Present address: Curie Institute, CNRS
UMR 168, Paris, France.}
and
Ken R. Duffy\thanks{Hamilton Institute, Maynooth University, Ireland.
E-mail: Corresponding ken.duffy@nuim.ie}
}

\date{$3^{{\rm rd}}$ June 2015}

\maketitle

\begin{abstract}
For proliferating cells subject to both division and death, how can
one estimate the average generation number of the living population
without continuous observation or a division-diluting dye? In this
paper we provide a method for cell systems such that at each division
there is an unlikely, heritable one-way label change that has no
impact other than to serve as a distinguishing marker. If the
probability of label change per cell generation can be determined
and the proportion of labeled cells at a given time point can be
measured, we establish that the average generation number of living
cells can be estimated. Crucially, the estimator does not depend
on knowledge of the statistics of cell cycle, death rates or total
cell numbers. We validate the estimator and illustrate its features
through comparison with published data and physiologically parameterized
stochastic simulations, using it to suggest new experimental designs.
\end{abstract}

\section{Introduction}
\label{sec:intro}

Given a proliferating population of cells starting from one or more
progenitors, a natural quantity of interest in cell biology is the
average number of divisions per cell since the initial progenitors,
i.e. the average generation of presently living cells. The average
generation is related to the population's turn-over rate and can
potentially be used to quantify dynamics and aging of the immune
system
\cite{Vaziri1993Loss,Kaszubowska2008Telomere,deboer2013,marchingo2014}, to
better understand the evolution and risk of cancer
\cite{Frank2003Patterns,Merlo2006Cancer,tomasetti15}, and to rank
cell types in hierarchies of complex differentiation programs
\cite{hills2009,zhang2013}.

Estimating average generation is a simple matter if cell lifetimes
are all equal and the division time is known or, alternatively, if
total cell counts can be measured and there is no cell death. If,
however, lifetimes are heterogeneous, the population is subject to
death as well as division or total cell counts are not available,
the issue is more involved. In these settings, several experimental
techniques can be employed to estimate average generation, including
time lapse microscopy, division diluting fluorescent dyes, and
inference from somatic mutations and telomere length.

The most unambiguous measurement technique is \textit{in vitro} time
lapse microscopy as it affords nearly direct determination of lineage
trees and cell generation. Time lapse microscopy has been used to
study many cell systems including bacteria, lymphocytes, embryonic
development, gut development, to name but a few, e.g.
\cite{Powell55,Smith73,Sulston83,Hawkins09,Snippert10,markham10,Gomes11,Duffy12,Giurumescu2012Quantitative,richards2013,dowling2014}.
Even it, however, has limitations as filming is not continuous,
cells can leave the field of vision or form three dimensional
structures that inhibit tracking. The complexity of image segmentation
increases with cell numbers, and so time lapse microscopy has proved
challenging if more than approximately $10$ generations are to be
followed.

Another popular experimental approach, particularly in immunology,
is to stain starting cells with a fluorescent dye such as
Carboxyfluorescein Succinimidyl Ester \cite{lyons94,lyons00,hawkins07},
CellTrace$^\text{TM}$ Violet or Cell Proliferation Dye eFluor 670
\cite{quah12}. With each division, cells inherit approximately
half of their parent's dye and so fluoresce with half their intensity.
A cell's generation can thus be determined by its luminous intensity
via flow cytometry. This approach
is used both \textit{in vitro} and \textit{in vivo}, and allows the
experimenter to start with a large number of progenitors without
difficulty. It enables 6-10 generations to be followed before dye
dilution is such that the signal-to-noise ratio is too high for
a cell's generation to be reliably determined.

Determining generation \textit{in vivo} remains challenging as often
it cannot be achieved by direct observation or cell stain methods.
Estimating replicative history, cell depth and lineage trees has
been proposed by measurement of average telomere length
\cite{harley1990,allsopp1992,vaziri1994,weinrich1997,rufer1999,hills2009}
or by the number of somatic mutations, which are introduced
during DNA duplication
\cite{shibata96,tsao00,shibata06,Wasserstrom08,reizel11,carlson2012}.
While methods in this direction rely on inference rather than direct
observation, they offer the possibility of tracing more than 10
generations in a wide range of species, including humans.

In the present paper, we provide a novel
average generation estimator that is
designed with the \textit{in vivo} setting in mind. The estimator is
suitable for systems where cells undergo an unlikely, heritable
division-linked label change with a determinable probability, where
the label serves only as a marker and does not impact on cell
dynamics. Chief amongst the estimator's desirable properties are
that: (1) it requires no information on cell lifetime distributions;
(2) the population can be subject to death as well as division; and
(3) only a proportion of label-positive to total cells needs to be
measured. The present article introduces the estimator, analytically
establishes its fundamental properties, validates its applicability
by comparison with simulation and comparison with published data,
proposes an experimental realization, and demonstrates its utility.

In Section \ref{sec:estover} we describe the estimator and explain
how it can be used. The estimator appropriateness is a consequence
of theorems that are presented in detail in Section \ref{sec:math},
with a heuristic explanation and overview of them appearing first
in Section \ref{sec:mathsum}. 

In Section \ref{sec:val_sim}, Validation Using Simulated Data, we
use Monte Carlo simulations to assess the estimator's performance
for physiological parameterizations.  In Section \ref{sec:val_data},
Validation Using Published Data, we use several sources of publicly
available data to illustrate the estimator's applicability. We avail
of complete lineage tree data for the development of \textit{C.
elegans}, as determined from time-lapse microscopy \cite{richards2013},
stochastically labeling it. Applying the estimator results in
accurate inference of the average generation in comparison to the
directly observed quantity.  We also take data from two distinct
experimental studies, one on human colorectal cancer cells
\cite{Gasche2003} and one on murine embryonic fibroblasts
\cite{Kozar2013}, that employ micro-satellite mutation fluorescent
reporter systems. Micro-satellite mutation is an unlikely division
linked change and the fluorescence of cells in these systems serves
as a label suitable for average generation estimation. The results
show consistency between average generation estimates, measured
quantities and values reported in the literature.

As an illustrative example, in Section \ref{sec:expdes}, Experiment
Design, we propose a genetic construct, based on existing pieces,
to facilitate average generation inference. We describe how the
probability of label-loss could be measured and how the method could
be validated. In Section \ref{sec:disc}, Discussion, we conclude
with experimental designs where the method would prove biologically
informative.

\section{Estimator Overview}
\label{sec:estover}

We consider a system where cells are subject to a division-linked,
heritable label change that serves as a measurable distinguishing
marker, but does not influence population dynamics. The method is
appropriate regardless of whether initial cells are label-negative
and can gain the label, which is then inherited by their offspring,
or are label-positive and can lose the label, with their offspring
not regaining the label. For a consistent description, we phrase
the paper in terms of the latter, but the results all hold {\it
mutatis mutandis}.

During each cell's lifetime, assume that a label-positive cell
becomes label-negative with a known, small likelihood, $p$.  Let
$Z(t)$ be the total live cell count at time $t$ and $\Zp(t)$ be the
live cell count of label-positive cells. Assume that the initial
cells at $t=0$ are all label-positive and that at some time
$t$ the fraction of label positive cells to total cells, $\Zp(t)/Z(t)$,
can be measured.  With $G(t)$ denoting the sum of the generations
of all cells living at time $t$ and with $G(0)$ defined to be $0$, then
given there are label-positive cells in the system at $t$, i.e. $\Zp(t)>0$,
we establish that the average generation of cells alive at time $t$
satisfies the following relationship 
\begin{align} 
\label{eq:approx}
\frac{G(t)}{Z(t)}
	\approx -\frac1p\log\left(\frac{\Zp(t)}{Z(t)}\right).
\end{align} 
That is, the average generation of the population can be estimated
directly from the proportion of label-positive cells if the delabeling
probability is known.

Perhaps unexpectedly
the formula \eqref{eq:approx}
does not depend on several difficult-to-measure factors such as
cell-lifetime distributions and total cell counts. Moreover the
right hand side of \eqref{eq:approx} requires only a proportional
measurement of label-positive cells, which can be determined
from a sample, and the relationship holds even though the population
could be subject to death as well as division.

For the validity of \eqref{eq:approx}, we have assumed that at $t=0$
all cells are label-positive. As a result $\Zp(0)=Z(0)$ and so both
the right and left hand side of \eqref{eq:approx} are $0$ at $t=0$,
agreeing irrespective of the initial starting number. If, however,
not all cells are initially label-positive, the estimator and the
average generation would not agree. This can be rectified
if measurements of the proportion of positively labeled cells are
available at two time-points, $t_2>t_1$.  Then, irrespective of the
initial composition of label-positive and label-negative cells, so
long as $\Zp(t_2)>0$,
\begin{align}
\label{eq:two_tps}
\frac{G(t)}{Z(t)} 
	\approx \left(\frac{t}{t_2-t_1}\right) \left(-\frac1p 
	\log\left(\frac{\Zp(t_2)Z(t_1)}{\Zp(t_1)Z(t_2)}\right)\right).
\end{align}

This two-measurement estimator has an additional advantage when
initial cell numbers are small. If the culture is started, as an
extreme example, with a single progenitor, then the right hand side
of \eqref{eq:approx} can be subject to stochastic fluctuations at
shorter time-scales (see the Monte Carlo simulations in Section
\ref{sec:val_sim}). As established rigorously in Section \ref{sec:math},
with $t_1>0$, the two-sample estimator in \eqref{eq:two_tps} is
more accurate than \eqref{eq:approx} as it removes the influence
of the timing of early cell events on the estimate.

\section{Explanation of the Estimator's Origin}
\label{sec:mathsum}

The estimators \eqref{eq:approx} and \eqref{eq:two_tps} have
non-obvious forms. Utilising expansion properties of cumulant
generating functions, results in the following section show that
the relationships hold, in expectation, for an arbitrary familial
relationship. This includes, in particular, estimation of average
generation for heterogeneous cell populations with distinct,
potentially interdependent, proliferation characteristics.

In the absence of a generative model of family trees, however, that
derivation cannot provide information about the development in time
of the average generation of a family of cells. Nor can it be used
to determine time-dependent properties of the estimators on a
developing population.  For a more detailed analysis of sample-path,
multi-progenitor and time-dependent properties, a general mathematical
framework for capturing the stochastic dynamics of a proliferating
cell system subject to division and death, as well as heterogeneous
cell life times, is that of age dependent branching processes
\cite{Harris63,Kimmel02} where cells have random lifetimes as well
as random numbers of offspring. Since the seminal work of Bellman
and Harris \cite{Bellman52}, it has been known that if cells are
more likely to divide then die, given the cell population does not
die out, the number of cells living at time $t$ grows exponentially
in time, $Z(t) \sim \exp(\alpha t)$, at a rate, $\alpha$, dubbed
the Malthus parameter, that depends on the lifetime and offspring
number distributions of cells. This result is known to be robust,
for example, to sibling dependencies \cite{crump69,Olofsson96,duffy09}.

In a cell system experiencing heritable one-way label-changes,
label-positive cells can become label negative-cells, but the reverse
is not possible. Thus the number of label-positive cells at time
$t$ also satisfies the branching process result, $\Zp(t) \sim
\exp(\alphap t)$, but with a Malthus parameter that is smaller than
that for the total label-independent population of cells,
$\alphap<\alpha$. The difference $\alpha-\alphap$ depends on the
likelihood, $p$, that a label-positive cell loses its label.

Fewer results have been established regarding the total generation
of a cell population, the sum of the generations of all cells living
at time $t$, $G(t)$, for which new theorems can be found in Section
\ref{sec:math}. For the label-independent population, in substantial
generality we prove that $G(t) \sim t\exp(\alpha t)$. That is, the
total generation increases faster than the total population size
by a factor of $t$.  Recalling that $Z(t)\sim \exp(\alpha t)$, for
a general age-dependent branching process, the average generation
of the population grows linearly in time $G(t)/Z(t)\sim g t$ for
some $g>0$. In the following section we give a deterministic result
to provide non-probabilistic intuition for why this is so before
considering the stochastic system.

The clinching result quantitatively relates the phenomena of
stochastic delabeling to generational expansion. So long as
label-positive cells remain
\begin{align*}
\log\left(\frac{\Zp(t)}{Z(t)}\right)
\stackrel{t \text{ large}}{\sim} \log(\exp((\alphap-\alpha)t))
= (\alphap-\alpha)t \stackrel{p\text{ small}}{\sim} - pgt,
\end{align*}
so that for $p$ small
\begin{align*}
\frac{G(t)}{Z(t)} \sim -\frac1p \log\left(\frac{\Zp(t)}{Z(t)}\right).
\end{align*}
This identification holds for any age dependent branching process
and does not depend upon the details of life-time distributions or
the possibility of cell death, so long as populations do not go
extinct.

The estimators are related through stochastic quantities and are
subject to sample-path fluctuations, particularly at early time-points.
Establishing that estimates, as a function of time on individual
sample paths, converge to the true average generation underlies the
development of the two time-point estimator in equation \eqref{eq:two_tps}.
That methodology circumvents this issue of small number variability
by eliminating the early fluctuations on a path-by-path basis. As
a supporting result, we also prove that the path-to-path variability
of estimates decreases inversely proportionally to the number of
progenitors, supporting the precision of estimates for experimental
systems that are initiated with multiple cells, which is typically
the case.

\section{Formal Results}
\label{sec:math}

\subsection{In-expectation derivation}

We begin by deriving a version of \eq{eq:approx} based on averages
over realisations of the delabeling process. This derivation has
the advantage that it requires no assumptions regarding the family
tree and so holds in complete generality, but it is not informative
with regards individual realisations. Sample-path, time-dependent
properties of the estimators arise as a consequence of involved
theorems in the context of the most well established generative
model of family trees, age dependent branching processes, which
follow this derivation.

At some time $t$, consider a collection of $Z(t)$ cells, whose
familial relationship need not be known. The generation of a cell
is defined to be the number of ancestors it possesses, with initial
cells being 
defined to
be in generation $0$. Denote each of the
generations of the $Z(t)$ cells by $g_1(t),\ldots,g_{Z(t)}(t)$.
of the generations of all living cells being denoted by
\begin{align*}
G(t)= \sum_{i=1}^{Z(t)} g_i(t),
\end{align*}
knowing $p$ we wish to estimate the average generation of presently
living cells, i.e.  $G(t)/Z(t)$, by observation of the proportion
$\Zp(t)/Z(t)$.

If initial cells are label positive and with each division the label
is lost independently with probability $p$, the probability cell
$i$ is still label positive in generation $g_i(t)$ is $(1-p)^{g_i(t)}$.
With $\Zp(t)$ denoting the number of label positive cells, by
linearity of expectation the average proportion of label-positive
cells in
the population is
\begin{align*}
        \frac{E(\Zp(t))}{Z(t)} 
        = \frac{1}{Z(t)} \sum_{i=1}^{Z(t)} (1-p)^{g_i(t)}.
\end{align*}
Identifying $\theta = \log(1-p)$, this can be re-written in
a form that identifies it with a moment generating function,
e.g. \cite{feller68},
\begin{align*}
        E\left(e^{\theta \Gamma}\right)=
        \frac{1}{Z(t)} \sum_{i=1}^{Z(t)} \exp\left(\theta g_i(t)\right),
\end{align*}
where $\Gamma$ is uniformly selected in $\{g_1(t),\ldots,g_{Z(t)}(t)\}$.
As both the moment generating function and the cumulant generating
function, $\log E(e^{\theta \Gamma})$, are real analytic, we can
take a Taylor expansion of the latter around the origin, giving
\begin{align*}
         \log\left(E(e^{\theta \Gamma})\right)
                = 0 + \theta E(\Gamma) + O(\theta^2),
\end{align*}
so that
\begin{align*}
\lim_{\theta\to 0} \frac{1}{\theta}
         \log\left(E(e^{\theta \Gamma})\right)
                =  E(\Gamma). 
\end{align*}
Taking $\theta\to 0$ is equivalent to taking $p\to 0$. Thus 
noting that $\lim_{p\to 0}p/\log(1-p)=-1$, we obtain 
\begin{align*}
\lim_{p\to 0} \frac{-1}{p}
        \log\left(\frac{E(\Zp(t))}{Z(t)}\right)
= \frac{G(t)}{Z(t)},
\end{align*}
which is formula \eqref{eq:approx}, albeit with an expectation over
stochastic delabelings. While this result is not as strong as others
we shall prove, it illustrates both how the unusual formulation
arises and that, averaged over delabeling processes, the relationship
holds for arbitrary family tree structure. 

\subsection{Sample path properties}

To determine sample path properties of the generation counting
estimator, we must establish two quantitative relationships in a
general age-dependent branching process: (1) relating the growth
rates of labeled and unlabeled populations when the probability of
delabeling per cell division is small; and (2) for the relationship
between the number of cells alive and the sum of the generations
of all living cells.

For the quantitative relationship between the growth-rates of the
labeled and unlabeled populations when the delabeling probability
per cell division is small and label-state does not impact on
population dynamics, we leverage well-know single type results
\cite{jagers1969,Athreya76}.
For the relationship
between the number of cells alive and the sum of the generations
of all living cells, there are some
results indirectly available from \cite{Samuels71}, but to identify
more properties of our estimator we find it beneficial to provide
an analysis of the joint probability generating function of number
of cells alive and their total generation number.

Our initial mathematical study treats systems that start with a
single initiating progenitor. From these results, systems that start
with several statistically equivalent progenitors with indistinguishable
progeny can be deduced. In addition, we provide results on the
behavior of the estimator for a system with numerous progenitors,
establishing that that variances behave inversely proportionally
to the number of progenitors.

\subsection{Smoothness of the Malthus parameter}

We begin by considering a standard age dependent branching process,
e.g. \cite{Harris63,Kimmel02}, with the usual assumptions on the
distribution of a lifetime, a non-negative random variable $\tau$,
and the number of offspring, $N$ a non-negative integer valued
random variable, at the end of a lifetime. In the circumstances of
interest to us, proliferating cell populations, $N$ will take values
in $\{0,2\}$ indicating that cells die or divide into two at the
end of their lives, but we will prove the results in greater
generality.

\begin{assumption}
\label{ass:1}
The lifetime distribution, $P(\tau\leq t)$ for $t\in[0,\infty)$,
is non-lattice and satisfies $P(\tau\leq0^+)=0$. The probability
generating function of the number of offspring, $\rho_N(s)=E(s^N)$
for $s\in\R$, is finite in a neighborhood of $1$. We denote its
expected value by
$\EN=E(N)=\frac{d}{ds}\rho_N(s)|_{s=1}$.
\end{assumption}

A key quantity in the study of age-dependent branching processes
is the Malthus parameter, $\alpha(\EN)$, which is defined to be the
solution of the following equation, should the solution exist,
\begin{align}
\label{eq:malthus}
\EN E(e^{-\alpha(\EN)\tau})=1.
\end{align}
If $\alpha(\EN)$ exists, it is then the, possibly negative,
asymptotic exponential growth rate of the expected population size.
The dependence of the Malthus parameter, $\alpha(\EN)$, on the
expected number of offspring of a cell, $\EN$, is not usually made
explicit, but will prove essential for us as we shall be interested
in relating the growth rate of the label-positive population and
of the total population, which will differ. For that purpose, we have
the following result on the range of mean offspring, $h$, for which
the Malthus parameter exists and its smoothness as a function of
$h$. We expect that this Proposition is known, but cannot find a
reference in the literature and so present a proof here.

\begin{proposition}
\label{prop:malder}
Define
\begin{align*}
\betamax :=\sup\left\{\beta:E(e^{\beta \tau}) <\infty\right\}
\text{  and  }
\hmin  := \inf\{\EN: \EN E(e^{\betamax\tau}) \geq 1\},
\end{align*}
where $\hmin\leq 1$,
then there exists a real analytic function 
$\alpha:(\hmin,\infty)\mapsto\R$, the Malthus parameter, such that
\begin{align*}
\EN E(e^{-\alpha(\EN)\tau})=1 \text{  for  } \EN\in(\hmin,\infty).
\end{align*}
In particular, for $h>\hmin$, the first two derivatives of $\alpha(\EN)$
as a function of the average number of offspring $h$ are
\begin{align}
\label{eq:maltder}
\alpha'(\EN):=\frac{d}{d\EN} \alpha(\EN) 
	=\frac{1}{\EN^2 E(\tau e^{-\alpha(\EN)\tau})}
\end{align}
and
\begin{align}
\label{eq:maltder2}
\alpha''(\EN):=\frac{d^2}{d\EN^2} \alpha(\EN) 
	= \frac{1}{\EN^3 E(\tau e^{-\alpha(\EN)\tau})}
	\left(\frac{E(\tau^2 e^{-\alpha(\EN)\tau})}
		{\EN E(\tau e^{-\alpha(\EN)\tau})^2} -2\right).
\end{align}
\end{proposition}
\begin{proof}
Consider the function $g:(0,\infty)\times\R\mapsto[0,\infty]$ defined
by
\begin{align}
\label{eq:g}
g(\EN,\beta)=\EN E(e^{-\beta\tau})-1.
\end{align}
We wish to identify the range of $h$ such that $\alpha(\EN)$ exists
satisfying $g(\EN,\alpha(\EN))=0$, and to determine its smoothness
properties as a function of $\EN$. 
When $\beta=0$, $E(e^{-\beta\tau})=1$ and so $\hmin\leq1$. 
As $E(e^{-\beta\tau})$
is the moment generating function of a non-negative random variable,
$\tau$, evaluated at $-\beta$, it is a real analytic function on
the interior of the domain on which it is finite, i.e. for
$\beta\in(-\betamax,\infty)$. Therefore as $g(\EN,\beta)$ defined in
\eq{eq:g} is a product of real analytic functions, it is also real
analytic on $(0,\infty)\times(-\betamax,\infty)$. 
As $g(\EN,\beta)$ is a monotonic decreasing function of $\beta$,
for any $\EN\in(\hmin,\infty)$ there exists a unique $\beta$ such
that $g(\EN,\beta)=0$. Thus we can apply the real analytic 
version of the Implicit Function Theorem, e.g. Theorem 1.8.3 \cite{Krantz02a}.
This establishes that $\alpha(\EN)$ satisfying \eq{eq:malthus} exists
for $\EN\in(\hmin,\infty)$, that $\alpha$ is real analytic for
$h>\hmin$, and that its first derivative satisfies
\begin{align*}
\alpha'(h) = 
\left.
-
\frac{\partial}{\partial h} g(h,\beta)|_{(\EN,\alpha(\EN))}
\middle/
	\frac{\partial}{\partial \beta} g(h,\beta)|_{(\EN,\alpha(\EN))},
\right.
\end{align*}
which gives \eq{eq:maltder}. To obtain \eq{eq:maltder2}, one
differentiates \eq{eq:maltder}.
\end{proof}

The real analyticity of $\alpha(\EN)$ allows us to characterize
the difference in growth rates of two populations in terms of the
Taylor expansion of $\alpha(\EN)$. This will be useful as if
there is small probability of label-loss, then the label-positive
and total populations have similar, but non-identical, average
offspring number. We will relate the difference in their growth
rates to the dynamics of the average generation of the population.

\subsection{One way labeling populations}

We consider a two-type reducible age-dependent branching process
previously studied, for example, in \cite{Jagers69}. That is, we
consider a cell system that starts with a single progenitor that
is positive for a label so that $\Zp(0)=1$ and $\Zn(0)=0$, and
define the total population at time $t$ to be $Z(t) = \Zp(t)+\Zn(t)$.
Each cell is assumed to have an independent and identically distributed
lifetime at the end of which they independently give rise to a
random number of offspring. Positive label cells can become negative
label cells, but the reverse does not happen. We are interested in
populations where the label does not indicate a phenotypic change
so that lifetime distributions do not depend on the label.

Depending upon the experimental system and marker employed, the
process of delabeling can occur in different ways and so one may
have distinct models. Thus we make a general assumption that
encompasses several of these.  Notationally, let the random variable
$N^+(p)$ define the number of label-positive offspring of a
label-positive cell, where $p\in(0,1)$ captures the probability of
delabeling. We will make a technical assumption that ensures that
the distribution of the number of label-positive offspring of a
label-positive cell is less than its total (i.e. label positive
and negative) offspring and that as $p$ drops to $0$ we assume that
we recover the random variable describing the total (i.e. label
independent) offspring of a cell, denoted $N:=N^+(0)$.

\begin{assumption}
\label{ass:2}
Lifetimes are independent and identically distributed, and independent
of independent and identically distributed type-dependent offspring
numbers. The lifetime distribution for cells, $P(\tau\leq t)$ for
$t\in[0,\infty)$,
is non-lattice and satisfies $P(\tau\leq0^+)=0$. If $p>0$, the
number of label-positive offspring of a label-positive cell is less
than its total number of offspring: $N^+(p)$ is strictly stochastically
dominated by $N=N^+(0)$, i.e. $P(N^+(p)\leq n)\geq P(N\leq n)$
for all $n\in\{0,1,\ldots\}$ and $P(N^+(p)\leq n)> P(N\leq n)$ for
some $n$. Thus, defining $E(N^+(p))=\ENplus$, $\ENplus<E(N)=\EN$
for all $p>0$. We assume that $\rho_N(s)=E(s^N)$ for $s\in\R$, is
finite in a neighborhood of $1$, that $\lim_{p\downarrow0}\ENplus=\EN$,
and that $\ENplus$ is real analytic in $p$.
\end{assumption}

{\it Example 1:} delabeling occurs immediately prior to a cell's
division, independently with probability $p$. In this case the
offspring of a label-positive cell's are either all labeled or all
delabeled. Thus with $N$ denoting the label-independent offspring
random variable,
\begin{align*}
E\left(s^{N^+(p)}\right) = p + (1-p)E\left(s^N\right) = p + (1-p)\rho_N(s)
\end{align*}
and, in particular, $\ENplus  = (1-p)\EN$.

{\it Example 2:} delabeling of the offspring of a label-positive cell
occurs independently with probability $p$ at birth. In this case
\begin{align*}
E\left(s^{N^+(p)}\right) = \sum_{n\geq0} s^n 
	\sum_{m\geq n} P(N=m) {m\choose n} (1-p)^np^{m-n},
\end{align*}
so that again $\ENplus = (1-p)\EN$, but higher order moments are
smaller than those in Example 1.

{\it Example 3:} the offspring of a label-positive cell experience
asymmetric label-loss. There are genetic constructs where, on
division, labels are lost asymmetrically \cite{zhang2013}. That is,
if a label-positive cell divides in two and one of its two offspring
loses its label, then the other does not. As a result, if $p$ is the
probability of delabeling and $q$ is the probability that a division
rather than death occurs, then
\begin{align*}
E\left(s^{N^+(p)}\right) = q\left(sp + s^2(1-p)\right),
\end{align*}
$E(N)=\EN = 2q$ and the mean number of label positive offspring is 
$E(N^+(p))=\ENplus=(1-p/2)\EN$.

In all three of these examples $\ENplus$ is a linear function of
$p$, so that Assumption \ref{ass:2} holds. Based on results in
\cite{jagers1969,Athreya76} in conjunction with Proposition
\ref{prop:malder}, the following theorem can be deduced.

\begin{theorem}
\label{thm:twopop}
If $\ENplus>1$, under Assumption \ref{ass:2},
if $\lim_{t\to\infty}\Zp(t)>0$, then almost surely
\begin{align}
\lim_{t\to\infty} \frac{1}{t} \log
	\left(\frac{\Zp(t)}{Z(t)}\right)
	&= \alpha(\ENplus)-\alpha(\EN)\nonumber\\
	&= p \ENplusprime \alpha'(\EN)
	+\frac{1}{2}p^2 \left(\ENplusprime^2
		\alpha''(\EN)+\alpha'(\EN) \ENplusprimeprime\right)
		+O(p^3) \label{eq:taylor},
\end{align}
where
\begin{align*}
\ENplusprime  = \ENplusprimefull
\text{ and }
\ENplusprimeprime  = \ENplusprimeprimefull,
\end{align*}
and, in particular, if $\lim_{t\to\infty}\Zp(t)>0$, then
almost surely
\begin{align}
\lim_{p\to0}\lim_{t\to\infty} \frac{1}{pt} \log
	\left(\frac{\Zp(t)}{Z(t)}\right)
	&= \ENplusprime\alpha'(\EN).
	\label{eq:taylor2}
\end{align}
\end{theorem}
\begin{proof}
The asymptotic behavior of $\Zp(t)$ and $Z(t)$ follow directly
from \cite{jagers1969,Athreya76}. By Proposition \ref{prop:malder}, $\alpha(\EN)$
is real analytic for $h>\hmin$ and $1\geq\hmin$, and, by Assumption
\ref{ass:2}, $\ENplus$ is real analytic with
$\lim_{p\downarrow0}\ENplus=\EN$. Thus $\alpha(\ENplus)$ is real analytic
and we can take a Taylor expansion
of $\alpha(\ENplus)$ in $p$ around $0$ and so \eq{eq:taylor} holds.
The relationship in \eqref{eq:taylor2} follows from \eq{eq:taylor}
taking the limit as $p$ tends to $0$.
\end{proof}

For comparison with these two-population results, we need to establish
properties of the total generation, the sum of the generations of
living cells, for a single-type population.

\subsection{Total generation of a branching process}

In order to precisely define the total generation of a branching
process, some unwieldy definitions are necessary. These shall be
largely side-stepped in the analysis, but are included to precisely
define the objects of interest to us.  We use a modification of
ideas in \cite{Harris63} to do so and consider a single-type
construction as that is sufficient for our needs.

The generation of an object is defined to be the number of
ancestors it possesses. As such, it is natural
to begin with a single object at time $0$ defined to be in
generation $0$. Let $\tau^{(i)}_j$ denote the length of the life
of the $j$-th object in the $i$-th generation and let $N^{(i)}_j$
denote the number of offspring of the $j$-th object in the $i$-th
generation. Starting with a single object,
the total number of objects that will, at any
stage, exist in the $n$-th generation is governed by a Galton-Watson
recursion:
\begin{align*}
K_0=1 \text{ and } 
	K_{n+1}=\sum_{k=1}^{K_n} N^{(n)}_k \text{ for } n\geq1, 
\end{align*}
where the empty sum is defined to be $0$. That is, the number of
objects that will ever be in generation $1$ is the offspring number
of the first object, $K_1=N^{(0)}_1$.  The number of offspring that
will ever exist in generation $2$ is the sum of the number offspring
of those members of generation $1$, and so forth, leading to the
above recursion for $K_n$.

Each object in generation $n$ is defined by a life-history,
$<i_0,\ldots,i_n>$, where this object is the $i_n$-th child of the
$i_{n-1}$-child of the $i_{n-2}$-th child and so forth back to
$i_0=0$. The indices run over objects that exist, so that
$i_n\in\{1,\ldots,K_n\}$. The birth time of the object $<i_0,\ldots,i_n>$
is the sum of lifetime of its ancestors,
\begin{align*}
T_{<i_0,\ldots,i_n>} = \sum_{k=0}^{n-1}\tau^{(k)}_{i_k}.
\end{align*}
The time the object ceases to be is $T_{<i_1,\ldots,i_n>} + \tau^{(n)}_{i_n}$.
In terms of these definitions, the total number of objects present
at time $t$ can be identified as all objects of any generation that
are alive at time $t$, 
\begin{align}
\label{eq:Zdef}
Z(t) = \sum_{n=0}^\infty \sum_{<i_0,\ldots,i_n>} 
	\chi\left(T_{<i_0,\ldots,i_n>}\leq t < T_{<i_0,\ldots,i_n>} + \tau^{(n)}_{i_n}\right)
\end{align}
where $\chi$ is the indicator function. That is, the population
alive at time $t$ is all those objects whose birth time is before
or equal to $t$ and whose death time is after $t$. 

The generation of an object $<i_0,\ldots,i_n>$ is $n$ and thus the
sum of the generations of all objects existing at time $t$ is
\begin{align}
\label{eq:Gdef}
G(t) = \sum_{n=0}^\infty n \sum_{<i_0,\ldots,i_n>} 
	\chi\left(T_{<i_0,\ldots,i_n>}\leq t < T_{<i_0,\ldots,i_n>} + \tau^{(n)}_{i_n}\right).
\end{align}
That is, for each object alive at time $t$ in generation $n$, the
total generation at that time is incremented by $n$. 

Before considering the probabilistic system, we begin with a
fundamental, non-probabilistic lemma that will give an indication
as to why one expects that the total generation of all living cells
$G$ to grow in a similar fashion to $Z\log Z$, where $Z$ is the
number of presently living cells, assuming that cells can either
die or divide into two at the end of their lives.

\begin{lemma}
\label{lem:min}
For a binary family tree beginning with a single progenitor in generation
$0$, if there are $Z>0$ cells alive then
the average generation per cell satisfies $G/Z\geq \log_2 Z$.
\end{lemma}
\begin{proof}
Consider a family tree with $Z = 2^n+m$ cells alive (i.e. $Z$ leaves), where
$n\in\{0,1,\ldots\}$ and $m\in\{0,1,...,2^{n+1}-2^n\}$. The minimal
$G$ given $Z$ is attained is when $2^n-m$ cells are in generation
$n$ and $2m$ cells are in generation $n+1$, which can be 
formally established using the equations given above. Thus
\begin{align*}
Z\log_2 Z=(2^n+m)\log_2(2^n+m)     
\text{ and }
        G = (2^n-m) n + 2m(n+1).
\end{align*}
If $m=0$ or $m=2^n$, these are equal and so the difference between
the two is zero at each $2^n$ and one needs to be delicate with
one's estimates. A rearrangement of terms gives $Z\log_2 Z\leq G$
if and only if
\begin{align*}
        (2^n+m) \log_2(1+m/2^n) \leq 2m. 
\end{align*}
A sufficient condition for this to be true is to relax the problem,
letting $m = x2^n$ with $x\in[0,1]$, and establish that
\begin{align*}
        (1+x)\log_2(1+x) \leq 2x \text{ for all } x \in[0,1], 
\end{align*}
which is readily achieved by calculus.
\end{proof}

In a system with death, in general there is no equivalent upper
bound to Lemma \ref{lem:min} as it's possible that there is only
one cell alive, $Z=1$, and the generation of that cell, and hence
$G$, can be arbitrarily large.  For branching processes, however, under
general conditions the population either dies out or ultimately
grows in a somewhat regular fashion, \cite{Harris63}, and so Lemma
\ref{lem:min} anticipates that, in a suitable sense, $G(t)/(tZ(t))$
becomes constant if the population survives. Indeed, amongst other
results, this is something we shall establish.

It is not possible to study the total generation of living cells,
$G(t)$, separately from the total population, $Z(t)$, as they are
strongly coupled and their dynamics are intertwined. This relationship
appears as a cross term in \eq{eq:integral} when we consider the
evolution of the joint probability generating function of the two,
giving rise to an integral equation of unusual form. Despite this
form, it is, in general, susceptible to analysis.

\begin{assumption}
\label{ass:3}
Lifetimes are independent and identically distributed and independent
of independent and identically distributed offspring
numbers. The lifetime distribution, $P(\tau\leq t)$ for $t\in[0,\infty)$,
is non-lattice and satisfies $P(\tau\leq0^+)=0$. The probability
generating function of the number of offspring, $\rho_N(s)=E(s^N)$
for $s\in\R$, is finite in a neighborhood of $1$. We denote its
first and second derivative at one by
\begin{align}
\label{eq:meanN}
\EN=E(N)=\frac{d}{ds}\rho_N(s)|_{s=1}
\text{  and  }
\ENtwo=E(N^2)-\EN=\frac{d^2}{ds^2}\rho_N(s)|_{s=1}.
\end{align}
\end{assumption}

\begin{theorem}
\label{thm:Ggen}
Under Assumption \ref{ass:3}, starting with $(Z(0),G(0))=(1,0)$,
the joint probability generating function of $\{(G(t), Z(t))\}$,
\begin{align*}
\rhogz(\sg,\sz,t) 
	:= E\left(\sg^{G(t)} \sz^{Z(t)}\right),
\end{align*}
satisfies the integral equation
\begin{align}
\label{eq:integral}
\rhogz(\sg,\sz,t) 
	= P(\tau>t)\sz
		+\int_0^t dP(\tau\leq u) \rho_N\left(\rhogz(\sg,\sg\sz,t-u)\right).
\end{align}
For $\EN>\hmin$, define 
\begin{align}
\label{eq:ch}
	c(\EN) = \ch{\EN},
\end{align}
where $\alpha(\EN)$ and $\alpha'(\EN)$ are defined in the statement of
Proposition \ref{prop:malder}.
For $\EN>\hmin$, 
\begin{align}
\label{eq:meanG}
\lim_{t\to\infty} 
\frac{E(G(t))}{te^{\alpha(\EN)t}} = c(\EN) \EN \alpha'(\EN),
\end{align}
and 
\begin{align}
\label{eq:GoverZ}
\lim_{t\to\infty} 
\frac{E(G(t))}{tE(Z(t))} = \EN \alpha'(\EN).
\end{align}
If $\EN>1$,
with $\ENtwo$ defined in \eq{eq:meanN},
defining
\begin{align}
\label{def:kappa}
\kappa=
\frac{\ENtwo E(e^{-2\alpha\tau})}
                {1- \EN E(e^{-2\alpha\tau})},
\end{align}
in addition we have that
\begin{align*}
\lim_{t\to\infty} 
\frac{E(G(t)Z(t))}{t E(Z(t))^2}
= \EN \alpha'(\EN) \kappa 
\text{ and }
\lim_{t\to\infty} 
\frac{E(G(t)^2)}{t^2E(Z(t))^2}
= \left(\EN \alpha'(\EN)\right)^2 \kappa.
\end{align*}
\end{theorem}
\begin{proof}
We first establish that \eq{eq:integral} holds.
Begin with a family tree starting with a single cell
at time $0$, $Z(0)=1$, defined to be in generation $0$, so that 
$G(0)=0$. Consider the impact of altering this initial condition
so that the founding cell is in generation $1$ at $t=0$, $G(0)=1$.
We will write $G_1(t)$ for the process counting the total generation
of living cells starting with $G(0)=1$, so that, in particular, $G_1(0)=1$,
and retain $G(t)$ for total generation given $G(0)=0$.  
The key observation,
which can be established
using equations \eq{eq:Gdef} and \eq{eq:Zdef}, is that
\begin{align*}
G_1(t) &= \sum_{n=1}^\infty (n+1) \sum_{<i_1,\ldots,i_n>} 
	\chi\left(T_{<i_1,\ldots,i_n>}\leq t < T_{<i_1,\ldots,i_n>} + \tau^{(n)}_{i_n}\right)
	=G(t)+Z(t).
\end{align*}
This can be understood by noting that, leaving all other aspects
of the family tree unchanged, the move from $G(0)=0$ to $G(0)=1$
causes the generation of every living cell to increase by $1$.

Thus, taking expectations over the lifetime of the initial cell
and using the independence and identical distribution of the sub-trees
generated by the first birth, but noting they start with cells that
are in generation $1$, we obtain
\begin{align*}
\rhogz(\sg,\sz,t) 
	&= P(\tau>t)\sz 
		+\int_0^t dP(\tau\leq u) \rho_N\left(
		E\left(\sg^{G_1(t-u)}\sz^{Z(t-u)}\right)\right)\\
	&= P(\tau>t)\sz 
		+\int_0^t dP(\tau\leq u) \rho_N\left(
		E\left(\sg^{G(t-u)+Z(t-u)}\sz^{Z(t-u)}\right)\right)\\
	&= P(\tau>t)\sz 
		+\int_0^t dP(\tau\leq u) \rho_N\left(
		E\left(\sg^{G(t-u)}(\sg\sz)^{Z(t-u)}\right)\right)\\
	&= P(\tau>t)\sz 
		+\int_0^t dP(\tau\leq u) \rho_N\left(
		\rhogz(\sg,\sg\sz,t-u)\right),
\end{align*}
giving \eq{eq:integral}. The uniqueness of $\rhogz(\sg,\sz,t)$
for $\sg<1$ and $\sz<1$ follows from analogous arguments to those
in \cite{Harris63}, Chapter VI, utilizing the ideas in \cite{Levinson60}.

To establish \eq{eq:meanG}, note that using \eq{eq:integral}
\begin{align*}
E(G(t)) 
= \frac{\partial \rhogz(\sg,\sz,t)}{\partial {\sg}}|_{\sg=\sz=1}
&= \int_0^t dP(\tau\leq u) \EN E(G(t-u))+\int_0^t dP(\tau\leq u) \EN E(Z(t-u))\\
&= \int_0^t dP(\tau\leq u) \EN E(G(t-u)) + E(Z(t))-P(\tau>t),
\end{align*}
where we have used the equivalent, well known integral equation for
$E(Z(t))$ to obtain the final equality. We wish to study the
asymptotics of this equation, dividing it by $t\exp(\alpha(\EN)t)$
and taking limits as $t\to\infty$.  To this end, we employ a result
from renewal theory, \cite{Asmussen98} Theorem 6.2(a), 
using the fact that by construction $\EN\exp(-\alpha(\EN)t)dP(\tau\leq t)$
is a probability measure,
\begin{align*}
\lim_{t\to\infty} \frac{E(G(t))}{te^{\alpha(\EN)t}} = 
\frac{1}{hE(\tau e^{-\alpha(\EN)\tau})} 
\left(\lim_{t\to\infty} e^{-\alpha(\EN)t} E(Z(t))
-\lim_{t\to\infty} e^{-\alpha(\EN)t} P(\tau>t)
\right),
\end{align*}
should the limits on the right hand side exist. The existence
of the first limit follows from results in \cite{Harris63},
\begin{align}
\label{eq:Zc}
\lim_{t\to\infty} e^{-\alpha(\EN)t} E(Z(t)) = c(\EN).
\end{align}
Consider $e^{-\alpha(\EN)\tau} P(\tau>t)$. If $\alpha(\EN)>0$,
then the limit automatically exists and is $0$. If $\alpha(\EN)<0$,
select $\beta>-\alpha(\EN)>0$ such that $E(\exp(\beta\tau))<\infty$.
From Proposition \ref{prop:malder}, we are guaranteed the existence
of such a $\beta$ as $h>\hmin$.  As $\beta>0$, as in the Chernoff
bound we have that
\begin{align*}
P(\tau>t)=E\left(1_{\{\tau>t\}}\right)\leq e^{-\beta t} E(e^{\beta\tau})
\text{ and hence }
\limsup_{t\to\infty} e^{\beta t} P(\tau>t)\leq E(e^{\beta\tau}).
\end{align*}
As $\beta>-\alpha(\EN)$, this ensures that
\begin{align*}
\lim_{t\to\infty} e^{-\alpha(\EN) t} P(\tau>t)=0.
\end{align*}
Hence, using \eq{eq:maltder},
\begin{align*}
\lim_{t\to\infty} \frac{E(G(t))}{te^{\alpha(\EN)t}} = 
	c(\EN)\EN \alpha'(\EN),
\end{align*}
as required. To obtain \eq{eq:GoverZ}, one uses \eq{eq:Zc} 
in conjunction with \eq{eq:meanG}.

Consider now, with a little re-arranging of terms,
\begin{align}
E(G(t)Z(t))
	=&\frac{\partial^2}{\partial\sz\partial\sg} 
		\rhogz(\sg,\sz,t)|_{\sg=\sz=1}\nonumber\\
	=&
	\int_0^t dP(\tau\leq u)
		\EN E(G(t-u)Z(t-u))
	+\int_0^t dP(\tau\leq u)
		\EN E(Z(t-u)^2)\nonumber\\
	&+\int_0^t dP(\tau\leq u)
		\ENtwo E(Z(t-u))^2
	+\int_0^t dP(\tau\leq u)
		\ENtwo
		E(G(t-u))E(Z(t-u)). \label{eq:EGZ}
\end{align}
We wish to establish that the limit
\begin{align*}
\lim_{t\to\infty}
	\frac{E(G(t)Z(t))}{te^{2\alpha(\EN)t}}
\end{align*}
exists and identify its limit. Using \eq{eq:Zc} and 
\begin{align*}
\lim_{t\to\infty} e^{-2\alpha(\EN)t}\,
	E((Z(t)^2))
	&= c(\EN)^2 \kappa,
\end{align*}
which can be found in \cite{Harris63}, it is clear that only the
last term on the right hand side of \eqref{eq:EGZ} could be non-zero.
As $e^{-2\alpha(\EN)t} dP(\tau\leq t)$ is a defective probability
measure, 
that is $\int_0^\infty e^{-2\alpha(\EN)t} dP(\tau\leq t)<1$,
an application of the Renewal Theorem for Defective Measures
\cite{Resnick1992}
in conjunction with the Dominated
Convergence Theorem leads us to
\begin{align*}
\lim_{t\to\infty} 
	\frac{E(G(t)Z(t))}{te^{2\alpha(\EN)t}}
= &\frac{\ENtwo}{1-\EN E(e^{-2\alpha(\EN)\tau})}
  \lim_{t\to\infty} 
	\int_0^t e^{-2\alpha(\EN)u}dP(\tau\leq u)
		\frac{E(G(t-u))}{te^{\alpha(\EN)(t-u)}}
		\frac{E(Z(t-u))}{e^{\alpha(\EN)(t-u)}}\\
= & \ENtwo
	\frac{E(e^{-2\alpha(\EN)\tau})}{1-\EN E(e^{-2\alpha(\EN)\tau})}
	c(\EN)^2 \EN \alpha'(\EN)\\
= & c(\EN)^2\EN \alpha'(\EN) \kappa.
\end{align*}

A similar argument reveals the other equality.
\end{proof}

From \eq{eq:taylor2} and \eq{eq:GoverZ}, in Theorems \ref{thm:twopop}
and \ref{thm:Ggen} respectively, we have our cornerstone result
relating the average generation of the population to the proportion
of positively labeled cells, justifying \eqref{eq:approx} and
\eqref{eq:two_tps}.

\begin{proposition}
\label{prop:result}
Under Assumptions \ref{ass:1} and \ref{ass:2}, if $\lim_{t\to\infty}\Zp(t)>0$,
then 
\begin{align*}
\lim_{t\to\infty} \frac{E(G(t))}{tE(Z(t))}
=\EN\alpha'(\EN) = 
\lim_{p\to0}\lim_{t\to\infty} \frac{1}{pt} \log
        \left(\frac{\Zp(t)}{Z(t)}\right)
\left(\frac{\EN}{\ENplusprime}\right)
\end{align*}
almost surely.
\end{proposition}

Consider the delabeling possibilities described after Assumption
\ref{ass:2}. In Examples 1 and 2, where delabeling occurs with
probability $p$ at the end of a lifetime or independently with
probability $p$ for each offspring, $\EN/\ENplusprime=-1$ and we
have that
\begin{align*}
\lim_{t\to\infty}\frac{E(G(t))}{tE(Z(t))}
	=
	-\lim_{p\to0}\lim_{t\to\infty} \frac{1}{pt} \log
	\left(\frac{\Zp(t)}{Z(t)}\right),
\end{align*}
almost surely, if $\lim_{t\to\infty}\Zp(t)>0$. Whereas, as in Example
3, if delabeling occurs via an asymmetric division construct,
$\EN/\ENplusprime=-2$ so that
\begin{align*}
\lim_{t\to\infty}\frac{E(G(t))}{tE(Z(t))}
	= 
	-2\lim_{p\to0}\lim_{t\to\infty} \frac{1}{pt} \log
	\left(\frac{\Zp(t)}{Z(t)}\right).
\end{align*}
almost surely if $\lim_{t\to\infty}\Zp(t)>0$. 

We have proved results for $E(G(t))/E(Z(t))$, experimentally
we would wish to know $E(G(t)/Z(t))$, the average generation per
progenitor. Thus we wish to establish that these results are a good
approximation for the latter. To achieve this we use the standard
approach, e.g.  \cite{Stuart94}, of taking a Taylor expansion of
the ratio around the ratio of the expectations, and quantifying how
the expectation of its leading order terms behave. 
Using the first and second order terms of the Taylor expansion, we have
\begin{align*}
E\left(\frac{G(t)}{Z(t)}|Z(t)>0\right) 
&\approx
\frac{E(G(t))}{E(Z(t))}\left(1+\frac{E(Z(t)^2)}{E(Z(t))^2}\right)
	- \frac{E(G(t)Z(t))}{E(Z(t))^2}.
\end{align*}
Using the results in Theorem \ref{thm:Ggen}, if $h>\hmin$
\begin{align*}
\lim_{t\to\infty}\frac{1}{t} \left(
\frac{E(G(t))}{E(Z(t))}\left(1+\frac{E(Z(t)^2)}{E(Z(t))^2}\right)
	- \frac{E(G(t)Z(t))}{E(Z(t))^2}\right)
=\EN \alpha'(\EN)
=
\lim_{t\to\infty}\frac{E(G(t))}{tE(Z(t))}.
\end{align*}
Thus using the ratio of the expectations as an approximation
to the expectation of the ratio is reasonable.

We can also use the Taylor approximation to determine the
approximate behavior of the variance of $G(t)/Z(t)$. Namely,
to first order
\begin{align*}
&E\left(\left(\frac{G(t)}{Z(t)}\right)^2|Z(t)>0\right) 
-E\left(\frac{G(t)}{Z(t)}|Z(t)>0\right)^2\\
&\qquad\approx
\frac{E(G(t)^2)}{E(Z(t))^2}
	-\frac{2E(G(t) Z(t))E(G(t))}{E(Z(t))^3}
	+\frac{E(G(t))^2E(Z(t)^2)}{E(Z(t))^4}.
\end{align*}
Using the results in Theorem \ref{thm:Ggen}, if $h>1$,
this shows that the variance is becoming small,
\begin{align*}
&\lim_{t\to\infty} \left(
\frac{E(G(t)^2)}{t^2E(Z(t))^2}
	-\frac{2E(G(t)Z(t))}{tE(Z(t))^2}
	\frac{E(G(t))}{tE(Z(t))}
	+\left(\frac{E(G(t))}{tE(Z(t))}\right)^2
	\frac{E(Z(t)^2)}{E(Z(t))^2}
	\right) = 0. 
\end{align*}
Thus we expect the average generation of a population of cells
to be a stochastically well behaved process; a phenomenon we
will observe in simulations described later.

\subsection{Large numbers of progenitors}
\label{subsec:largeno}

Cultures are typically started with more than a single progenitor
and so one expects that the estimator's accuracy will improve by
laws of large numbers. Here, by means that have nothing to do with
branching processes, we establish that this is indeed the case and
that, in particular, the ratios are both Asymptotically Normal,
e.g. \cite{Serfling80},
and the variance decreases as one over the number of progenitors.

In order to state the result, recall that a sequence of random
variables $\{X_n\}$ is said to be Asymptotically Normal, e.g.
\cite{Serfling80}, with means $\{\mu_n\}$ and variances $\{\sigma_n^2\}$
if $\sigma_n>0$ for all $n$ sufficiently large and the sequence
$\{(X_n-\mu_n)/\sigma_n)\}$ converges in distribution to $N(0,1)$,
the Gaussian distribution with mean $0$ and variance $1$.

\begin{theorem}
\label{thm:largen}
Let $\{(\Zp_i(t),Z_i(t))\}$ be a bivariate i.i.d sequence of possibly
correlated random variables, representing the label-positive and
total cell population at time $t$ from progenitor $i$, and let
$(\Zp,Z)$ be an independent copy. If the probability generating
function of $(\Zp,Z)$ is finite in a neighbourhood of $(1,1)$, then
\begin{align*}
-\frac1p \log
	\left(
	\frac{\sum_{i=1}^n\Zp_i(t)}{\sum_{i=1}^nZ_i(t)} \right)
\end{align*}
is Asymptotically Normal with
\begin{align} 
\mu_n = -\frac1p \log \frac{E(\Zp)}{E(Z)} 
\text{ and }
\sigma_n^2 = \frac1n \frac{1}{p^2} 
\left(\frac{\Var(\Zp)}{E(\Zp)^2}-2\frac{\Cov(\Zp,Z)}{E(\Zp)E(Z)}+\frac{\Var(Z)}{E(Z)^2}
\right),
\label{eq:varZplusZ}
\end{align} 
where $\Cov(X,Y)=E((X-E(X))(Y-E(Y)))$ and $\Var(X)=\Cov(X,X)$.

Let $\{(G_i(t),Z_i(t))\}$ be a bivariate i.i.d sequence of possibly
correlated random variables, representing the total generation
and total cell population at time $t$ from progenitor $i$, and let
$(G,Z)$ be an independent copy. If the probability generating
function of $(G,Z)$ is finite in a neighbourhood of $(1,1)$, then
$\sum_{i=1}^n G_i(t)/(\sum_{i=1}^n Z_i(t))$
is Asymptotically Normal with
\begin{align}
\mu_n=\frac{E(G)}{E(Z)}
\text{ and }
\sigma_n^2 = \frac1n 
\left(
\frac{\Var(G)}{E(Z)^2}-2\frac{\Cov(G,Z)E(G)}{E(Z)^3}+\frac{\Var(Z)E(G)^2}{E(Z)^4} 
\right).
\label{eq:varGZ}
\end{align}
\end{theorem}
\begin{proof}
Both results start with an application of the Multi-variate Central 
Limit Theorem, e.g. \cite{Serfling80},
to
the partial sums of $\{(\Zp_i(t),Z_i(t))\}$ and $\{(G_i(t),Z_i(t))\}$
to establish their Asymptotic Normality.
Application of the 
Delta Method, the corollary on page 124 of
\cite{Serfling80},
with the function $g(x,y)=x/y$ then establishes
Asymptotic Normality of the ratios of the sums. An
additional application of the Delta Method
with the function $g(x)=-p^{-1}\log(x)$ establishes the Asymptotic
Normality of the logarithm of the ratio of sums.
\end{proof}

As a result of this theorem, we anticipate the variability in the
average generation number, as well as the variability in estimator,
to decrease significantly as the number of progenitors increases.

\section{Validation Using Simulated Data}
\label{sec:val_sim}

Here we use Monte Carlo simulations of age dependent branching
processes \cite{Harris63,Kimmel02} to validate the method, investigating
its appropriateness on individual stochastic sample paths for
parameter values of the order one would anticipate experimentally.

Define $\Zp(t)$ to be the number of label-positive live cells at
time $t$, $Z(t)$ to be the total number of living cells at time
$t$, $G(t)$ to be the sum of the generations of all living cells
at time $t$, so that $G(t)/Z(t)$ is the average generation of living
cells at time $t$. In this section each cell's lifetime is drawn
independently from the same distribution and, at the end of their
life, each cell independently either dies or gives rise to two
cells. We assume that label-positive cells delabel immediately
prior to division with probability $p$.

One of the estimator's primary features is that it can be used in
practice without knowledge of the lifetime distribution or death
rates. Thus we performed stochastic simulations over a
range of conditions, including lifetime distributions such as
lag-exponential, lag-log-normal and lag-gamma, and parameter values
such as $p\in(10^{-5},10^{-1})$, $t\in[0,300]$ hours and the average
number of offspring of a cell, $E(N)=\EN$, in $(0,2)$. These values
include those typically encountered in cell cycle experiments
\cite{Smith73,Hawkins09,Sun2012Direct,Sun2008Mitotic,Wasserstrom08,Weber2014} and
demonstrated the merits of the method, even for $t$ being relatively
small, at a few days, and $p$ of order $0.01$ per cell per generation.
To illustrate the insights gained from the simulation study, we
report on a specific, representative example, where lifetime
distribution is lognormal, which is found to be a good fit in
\cite{Hawkins09}. Figures \ref{fig:single_progenitor_gam} and
\ref{fig:mult_progenitor_gam} are the equivalents of
\ref{fig:single_progenitor} and \ref{fig:multiple_progenitors}, but
for delayed-gamma distributed cell lifetimes. The results for both
of these settings, as well as others that we have explored, are
qualitatively similar.

\subsection{Simulation parameterization}

For the lifetime distribution we assume a log-normal distribution
with parameter values based on data from long-term video microscopy
of murine B lymphocytes proliferating in response to CpG DNA
\cite{Hawkins09}, $\tau$ is log-normally distributed with mean $9.3$
hours and standard deviation $\sigma=2.54$ hours. We consider
populations that are, on average, expanding. 

For the probability of delabeling, we adopt estimates related to
existing division-linked labeling systems, namely Cre-induced mitotic
recombination \cite{Sun2008Mitotic,Liu2011Mosaic} and micro-satellite
mutation induced label activation
\cite{Slebos2002,Koole2013Versatile}. The probability of
Cre induced mitotic recombination has been estimated for a specific
construct as $10^{-2}$ per cell per generation \cite{Sun2008Mitotic},
while mutation rates in human micro-satellites were inferred to be
in the $(10^{-4},10^{-3})$ range per locus per generation
\cite{Sun2012Direct}.

Cell populations are simulated for $T=250$ hours. At the end of a
cell's lifetime, it divides into two cells with probability $0.8$
and dies with probability $0.2$, giving an average number of offspring
$h=1.6$. This leads to an average expansion by a factor of approximately
$10^6$ at $250$ hours and an average increase in the generation
number of more than 25. Note that 25 generations is far beyond what
current experimental method, such as continuous observation or
division-diluting dyes, are able to measure.

The mathematical results hold for sample population ratios where
label-positive cells persist.  If at some time $t\in(0,T]$ $\Zp(t)=0$,
we do not report the path from $t$ on as this represent samples for
which the experimenter would have no estimate beyond that time.

\subsection{Single progenitors: generation consistency, label
ratio variability at early times}

Systems that start with a single label-positive cell are subject
to the greatest variability at early times and thus this situation
is the one that poses the greatest difficulties for estimation.
While most experimental systems will start with more than a single
progenitor and we know from the mathematical results in Section
\ref{subsec:largeno} that they are better behaved, with variance
in estimates decreasing inversely proportional to starting cell
number, we begin by considering this most challenging circumstance.

Starting with a single label-positive progenitor at $\Zp(0)=1$,
panels (a) and (b) of \Fig{fig:single_progenitor} plot $15$
realizations of both the average generation $G(t)/Z(t)$ (orange
lines) and the estimate of it, $-1/p\log(\Zp(t)/Z(t))$ (blue lines).
The average generation across realizations is consistent, even for
early times, and grows linearly in time as Theorem \ref{thm:Ggen}
predicts.

In contrast, the per-realization estimates, the blue lines, which
the sample path theory establishes will
ultimately grow with the same slope
as the average orange line, exhibit substantial variability at short
time-scales. Two dynamical properties contribute to these fluctuations.
At the beginning, there are no label-negative cells and so the
logarithm of the ratio is $0$.  When the first label-negative cells
appear, the ratio of label-positive to total cell population can
dramatically change, especially if the population size is small.
In effect, the theorem is not in force until both the label positive
and label negative populations are large enough for average behavior
to become dominant,
which only happens after at least an order 
of $1/p$ 
division events have occured.
As one might expect, by decreasing $p$ from
$10^{-2}$ to $10^{-3}$ this effect is amplified, as can also be
seen in \Fig{fig:single_progenitor} (a) and (b).

In this setting, the typical behavior is initially for underestimation
as can be seen in the inter-quartile range plot of panels (c) and
(d) of \Fig{fig:single_progenitor} for $10^4$ simulations. 
This can be understood as there is no estimate prior to at least one
cell delabeling, which requires, on average, $1/p$ divisions
to occur. As $p$ is assumed small, typically one does not intially
get an estimate and this delay explains the observed lag in the
blue lines. On the other hand, in the unlikely event that an early
cell delabels, then the approximation \eq{eq:approx} initially gives
signifcant over-estimates.
This bias, which we have observed consistently across other
simulations, suggests that if starting with a single progenitor
family, \eq{eq:approx} only becomes accurate at longer time-frames,
i.e. after several generations have passed. Note
that the mean value of the
estimator at early times over multiple runs is effected by outliers,
hence the mean lies outside the inter-quartile range.

\subsection{Single progenitors: mitigating variability by use of two
time-points}

Results in Section \ref{sec:math} for the time-dependent sample
paths of estimates suggest that one way to mitigate this initial
variability in the single time-point estimate $-1/p\log(\Zp(t)/Z(t))$
is to make two measurements at distinct times, $t_1<t_2$ and use
their difference to estimate $G(t)/Z(t)$ via the approximation
\eq{eq:two_tps} in Section \ref{sec:intro}:
\begin{align*}
\frac{G(t)}{Z(t)}
        \approx \left(\frac{t}{t_2-t_1}\right) \left(-\frac1p
        \log\left(\frac{\Zp(t_2)Z(t_1)}{\Zp(t_1)Z(t_2)}\right)\right).
\end{align*}
The effect of this difference measurement
is the removal of the
initial fluctuations on a path-by-path basis.
This is illustrated in \Fig{fig:single_progenitor}, panels (e) and
(f). Using Monte Carlo methods and a standard kernel estimation
method, the plots compare three values: the actual average generation
$G(t_2)/Z(t_2)$ at $t_2=250$ hours (orange); the single time-point
estimates $-1/p\log(\Zp(t_2)/Z(t_2))$ (blue); and, with $t_2=250$
and an extra measurement at $t_1=200$ hours, the two time point
estimates in \eq{eq:two_tps} with $t=t_2=250$ hours. The two-point
estimate is not only more symmetrically distributed around the
true value, but its variability is significantly less than the
single time-point estimates.

\subsection{Multiple progenitors: improved consistency}

While the method works well for single progenitors, experiments are
typically seeded with multiple progenitors. Theorem \ref{thm:largen}
establishes that the variability decreases inverse linearly with
the number of progenitors, $Z(0)$, and so one expects that starting
with even a small number of cells will eradicate the early time
variability and here we demonstrate that this is, indeed, the case.
For $\Zp(0)=100$, \Fig{fig:multiple_progenitors} panels (a) and (b)
correspond with those of \Fig{fig:single_progenitor}.  The orange
lines, showing the per-realization of the average generation, remain
largely unchanged, but the accuracy of the single time point estimates
(blue) are greatly improved.

This comes about as even if an individual progenitor's family
generates label-negative cells relatively late this is balanced by
another progenitor's family where delabeling happened relatively
early. Thus both $\Zp(t)$ and $Z(t)$ rapidly become sufficiently
substantial for the limit theorem to be in force and for the estimator
to be precise.

\Fig{fig:multiple_progenitors} panels (c) and (d) show Monte Carlo
created Box plots of the per-sample average generation, the single
time-point estimate and the two time-point estimate for a range of
initial progenitor numbers. Sample-to-sample variation is greatly
reduced even with as few as $10$ progenitors. Thus in standard
experimental systems, which are often seeded with hundreds or
thousands of cells, the estimator is expected to be precise.

\section{Validation Using Published Data}
\label{sec:val_data}

The simulations of the previous section show that the estimator
works well for parameterizations akin to those found in experimental
systems and provide insight into the impact of a number of factors,
such as the delabeling probability and the number of progenitors,
on the quality of the estimates. Here we use a range of published
data to further
investigate the estimator's applicability.

\subsection{Early development of \textit{C. elegans}}

The first application takes lineage trees of the early development
of \textit{C. elegans} determined by time-lapse microscopy and
reported on in \cite{richards2013}. As the entire family tree is
known, we have direct access to measurements of average generation
as a function of time. Simulating a stochastic delabeling process
on this tree, as would be experimentally possible via a division
linked mosaic construct similar to the one introduced in \cite{Koole2014}
or the one we propose in Section \ref{sec:expdes}, we can directly 
investigate the estimator's accuracy.

\Fig{fig:celegans} (a) shows
the embryonic lineage tree of
the nematode \textit{Caenorhabditis elegans} 
as
constructed from data 
published by Richards et al. \cite{richards2013}. In their study,
a reference lineage tree was constructed based on output from
automated tracking of nuclei in $18$ embryos, recorded over
approximately six hours at 1.5 minute intervals by three-dimensional
resonance-scanning confocal microscopy. Each node in the figure
represents a cell division and the tree contains information about
all parent-daughter relationships and gives the timing of more than
$10^3$ division events during an embryo's development. This leads
to a population size of approximately 600 cells, \Fig{fig:celegans} (b).
From the timing of the division events, lifetimes are readily computed,
shown in \Fig{fig:celegans} (c).

For this tree, lifetime durations correlate positively with both
generation and with the time of birth relative to fertilization,
\Fig{fig:celegans} (d), illustrating a lack of homogeneity. Regarding
independence, there are correlations throughout the tree.
\Fig{fig:celegans} (e) demonstrates, for example, that lifetimes
of siblings are positively correlated.
These features are not in line with the assumptions under which
some of the the sample path properties of the estimator were
established, but are consistent with the derivation via 
properties of the cumulant generating function in Section
\ref{sec:math}.

To test the accuracy of the estimator, we stochastically decorated
the tree. Beginning with a single label-positive progenitor, daughters
are delabeled with probability $p$. Should they remain label-positive,
their daughters are delabeled independently with probability $p$
and so on for the whole tree. \Fig{fig:celegans} (a) illustrates
one such random decoration, with blue indicating label positive and
white label negative.

For $10^4$ independent Monte Carlo realizations of this delabeling
process, \Fig{fig:celegans} (f) shows the average (blue) and the
inter-quartile range of the estimates (blue), using \eqref{eq:approx},
at each time point during the first six hours of the development
as well as the true average generation number of the embryo in
orange, proving to be accurate. 

As explained in Section \ref{sec:estover}, if not all cells are
initially label-positive one can estimate average generation from
proportion measurements at two distinct times. To mimic this, we
considered the development of the labeled trees after $t_1=150$
minutes. At $t_1=150$ minutes one measurement is taken, giving the
proportion of label-positive cells and the approximation \eq{eq:two_tps}
to determine the average generation growth since $t_1$ to the time
of a later measurement to determine the average generation since
$t_1$.  The results are plotted in \Fig{fig:celegans} (g), which
show accuracy.

\subsection{Micro-satellite mutation reporter systems}
Existing micro-satellite reporter systems implement division-linked
labeling and are suitable for both \textit{in vitro} and \textit{in
vivo} applications
\cite{Slebos2002,Gasche2003,Kozar2013,Koole2013Versatile,Koole2014}.
Micro-satellites are short repeating motifs found in DNA that, with
small probability per cell division, are subject to insertion or
deletion of copies of the 
motif. Mutation reporter systems possess
an initially out of frame gene for a fluorescent protein. Micro-satellite
mutation on division results in the gene becoming in frame, with
the cell becoming fluorescent. This fluorescent state is inherited
by offspring and is either irreversible, as in \cite{Koole2014},
or has a reportedly negligible likelihood of reversion
\cite{Gasche2003,Kozar2013}. Thus with micro-satellite mutation
acting as a driver of rare division linked change, cell fluorescence
serves as a label for average generation estimation.

The two data sets we analyze are from studies published by Gasche
et al. \cite{Gasche2003} on human colorectal cancer cells (HCT116)
and Kozar et al. \cite{Kozar2013} on mouse embryonic fibroblast
(MEF). The probability that a cell becomes fluorescent as a result
of a micro-satellite mutation is estimated as $6.1\times10^{-4}$
and $1.1\times10^{-4}$ in each study, respectively. For generation
estimation, these values serve as $p$, the probability of delabeling
per cell division.

The data sets consist of measurements of proportions of labeled
cells over several days from in vitro proliferating cell lines that
carry transgenic micro-satellite mutation reporter cassettes,
allowing us to estimate the average generation of the cell population.
In addition to the proportions, \cite{Kozar2013} also provides the
growth curve, which is exponential from day one to day four before
slowing down, probably as a consequence of cell confluence
\cite{Kozar2013}, while \cite{Gasche2003} reports an exponential
growth for the whole time-course of their experiment.

Applying the estimator $-1/p\log(\Zp(t)/Z(t))$ with the proportions,
$\Zp(t)/Z(t)$, recapitulated in \Fig{fig:MS} (a-b), and label loss
probabilities provided in the studies, we can estimate the average
generation $G(t)/Z(t)$ as a function of time, \Fig{fig:MS} (c). For
both data sets, the estimated average generation increases linearly
over time at a rate of approximately one per day.

Direct validation of the method with this data would require knowledge
of the average generation, which is not reported in the above
studies. The growth curve \Fig{fig:MS} (d), however, allows an
alternate method to estimate this quantity for comparison. We
directly estimate the average from this curve by assuming for
exponentially distributed life-times and the absence of cell death
as in this special case one can show that $E(G(t))/E(Z(t))=2\mu t$,
where $\mu$ is the Malthus parameter \cite{Harris63,Kimmel02}.
Estimating $\mu$ from the growth curve from day one up to day four,
using linear regression on $\log(Z(t))$, predicts that the increase
of average generation relative to the first time point as $1.06\:
t$. This predictions matches closely the estimated average generation,
\Fig{fig:MS} (c). Taken together, the analysis of these data shows
that inferring the average generation from the proportion of labeled
cells carrying micro-satellite mutation reporter cassettes is an
experimentally feasible approach.

\section{Experiment Design}
\label{sec:expdes}

As an illustrative example of an ideal experiment including estimation
of the probability of label-loss $p$, we describe one possible
implementation of a division-linked one-way labeling system via a
genetic construct, \Fig{fig:construct} (a), that combines existing
experimental techniques: expression of a fluorescent protein such
as Blue Fluorescent Protein (BFP); the use of a cell-cycle specific
promoter to create division-linked changes
\cite{sakaue2008,Bai2010Nucleosomedepleted}; and site-specific
recombination \cite{Nagy2000Cre}.

The construct is composed of two elements. First CRE recombinase
expression is placed under the control of a cell-cycle specific
promoter \cite{sakaue2008,Bai2010Nucleosomedepleted}. The cell-cycle
promoter is employed to ensure a division linked expression of CRE
recombinase. Then, two LoxP sites are placed at each end of the BFP
gene. As cells enter cell cycle CRE recombinase is expressed,
site-specific recombination will occur probabilistically between
the two LoxP recombination sites \cite{Nagy2000Cre} and the BFP
gene is excised.

For average generation inference, the desirable likelihood of
recombination should be small, which can be achieved with a low
efficiency CRE recombinase. Thus BFP+ cells are regarded as label
positive and BFP- cells are label negative.  This design is similar
in spirit to an existing one used to create mosaics in Zebrafish
\cite{Koole2014}, which employs micro-satellites as a division
linked probabilistic driver and a kaloop in lieu of site-specific
recombination.

The proportion of label-positive cells in a cell system incorporating
this construct can be readily determined via fluorescence-activated
cell sorting (FACS). 
The essential remaining ingredient for average generation inference
is the
determination of the probability of label-loss. This can be achieved
by a 
\textit{in vitro} FACS experiment, as illustrated by the
\textit{in silico} simulated experiment 
in \Fig{fig:construct} (b)-(c). A collection of cells
that incorporate the construct are stained with a division diluting
dye, such as CFSE \cite{quah12}, 
with noise added for illustrative purposes.
After division, cells
are gated based on their generation, and the proportion of
label-positive cells per-generation determined by their fluorescence.
With $\BFPp(n)$ being the measured proportion of label-positive
cells in generation $n$, in the presence of a large number of initial
progenitors the probability label-loss, $p$, irrespective of whether
deaths occur, is $1-\BFPp(n+1)/\BFPp(n)$, which should be the same
for every $n$. In particular, if all initial cells are BFP+, then
$\BFPp(0)=1$ and $p=1-\BFPp(1)=\BFPn(1)$, the proportion of BFP
negative cells in generation $1$. Thus the probability of label-loss
can be readily determined experimentally.  For this simulation, an
estimated value of $p$ is determined using this measurement for
$n=0$, $p=\BFPn(1)$.

The accuracy of the method can then be checked during or after the
experiment by comparing the average generation as determined by the
division diluting dye with the value estimated via equation
\eqref{eq:approx} with the measured $p$, \Fig{fig:construct} (d).
Once $p$ is known, similar cell stain experiments with distinct
cell lines or distinct stimulii can be performed to validate the
method's inference before its use \textit{in vivo} and for generation
counts beyond the range possible with division diluting dyes or
fluorescence microscopy.

The construct described in this section is intended as an exemplar,
creatable with present technology. Further experimental possibilities,
including those based on naturally occurring mutations, can be found
in the following Discussion.

\section{Discussion}
\label{sec:disc}

The estimators introduced in this article, \eq{eq:approx} and
\eq{eq:two_tps}, offer means to estimate a biologically significant
quantity, the average generation of cell populations, for 
cell systems where direct measurement poses a significant
challenge. Their appropriateness is not intuitively apparent and
their development requires several mathematical results. 
Throughout the development of those results we assume that the
label change is one-way, which can be shown not to be essential for
the cumulant generating function derivation of the estimators, but
is relied on in establishing the sample-path results as if
the label change is reversible, so that labels can be gained
as well as lost, then asymptotically it is known that the growth
rate of the number of label-positive cells equals the growth rate
of the total number of cells \cite{Mode66}. 

Despite its less than obvious genesis, 
the estimator has highly
desirable features: it allows the population to be subject to death
and division; it does not need to know the number of progenitors;
it does not require knowledge of cell-cycle distributions; and,
subject to knowing the per-division label-change probability, only
requires the measurement of proportions. Comparison of estimates
for simulations suggest that the estimator is accurate for physiological
reasonable parameterizations.

The estimators \eqref{eq:approx} and \eqref{eq:two_tps} 
depend on the proportion of label-positive cells in
the population, where the label can be any of the cell's properties
that is lost with a small probability during its division cycle.
Several naturally occurring and engineered cellular processes exist
that approximately match this requirement. 
Somatic mutation is probably the best studied natural process
of this kind. In this case germ-line cells are defined to be
label-positive, while label-negative cells correspond to those with
mutations in their DNA. The proportion of label-positive cells can
be assessed by next generation sequencing. In humans there is a
high rate of fidelity in DNA replication, giving an error rate of
$5 \times 10^{-11}$ per base per division \cite{Drake1998} and so
using only a single nucleotide location as a label appears 
impractical. Considering instead the full germ-line genome as a
positive label, 
even though each nucleotide may suffer mutations with distinct
rates, and potentially in concert, there is still an over-all
probability of label-change that
has been estimated in
human sperm cells as approximately $p=4 \times 10^{-3}$ \cite{Drake1998}.
In principle, this rate could be used as the probability of delabeling
for our estimator. Alternatively, hotspots of mutation could also
be analyzed. For example, micro-satellites, short repetitive sequences
of DNA, mutate in humans with $p$ in the range $(10^{-4},10^{-3})$,
which we have used for the parameterization to evaluate the performance
of our estimator.

The use of hotspots of mutation as a label
suffers from the difficulty that delabeled cells could relabel due
to the occurrence of further mutations. The impact of this can be
ameliorated by creating an asymmetry in the likelihood of relabeling
to delabeling. One selects a set of $N$ sites for measurement, each
having their own probability of mutation, $p_i$. One defines the
label-positive cells as those that have germ-line values at all
sites. The likelihood of a labeled cell delabeling is the likelihood
that any of the sites mutate away from germ-line, $p=1-\prod_{i=1}^N
(1-p_i) \approx \sum_{i=1}^N p_i$, if the individual $p_i$ are
small. The likelihood that a cell with a single mutation relabels
is $p_i \ll p$, creating a significant asymmetry. Moreover, for a
cell that has had several mutations, the likelihood that all revert,
which is necessary for relabeling, is smaller still.

While it seems likely that one must rely on naturally occurring
processes to estimate the average generation of a cell population
for humans, genetically modified cell lines and animals have the
advantage that measuring the proportion of label-positive cells can
be directly facilitated by fluorescent markers. The construct
described in the Experiment Design section could be built or existing
constructions, such as Cre-mediated sister chromatid recombination
\cite{Sun2008Mitotic,Liu2011Mosaic,zhang2013} and micro-satellite
mutation induced label activation
\cite{Slebos2002,Koole2013Versatile}, can be adapted,
so that reversion does not occur.
Using fluorescence-activated cell sorting or live imaging technologies
and the estimators developed in this article it 
could
be possible to follow the average generation of specific cell
populations, like neoplastic tissue or the skin, in living organisms
over long periods of time.

Recent evidence from \textit{in vivo} cell lineage tracing techniques such
as Cellular Barcoding \cite{golden1995} has demonstrated that
apparently homogeneous progenitors give rise to heterogeneous
families in cancer, immunology and hematopoesis, e.g.
\cite{lu2011,kreso2013,buchholz2013,gerlach2013,naik13}. Theoretical
work developing methodologies to interrogate data from these
experiments is ongoing, but one key difficulty that must be overcome
is that the generation of individual families is unknown \cite{perie14}.
Combining any of the above experimental methodologies and using the
estimator developed in the present article would allow inference
of the per progenitor average generation, enhancing the deductive
power of these techniques, which have already proved invaluable.
Having established the estimators and validating them through
simulations and comparison with published data, we believe these
are promising avenues.

{\bf Acknowledgments:}
The authors thank S{\o}ren Asmussen (Aarhus University) for drawing
their attention to \cite{Asmussen98}. The work of T.W., L.P. and
K.D. was supported by Human Frontier Science Program grant RGP0060/2012.
K.D. was also supported by Science Foundation Ireland grant 12 IP
1263.

\newpage

\begin{figure}
\centering\includegraphics[scale=0.7]{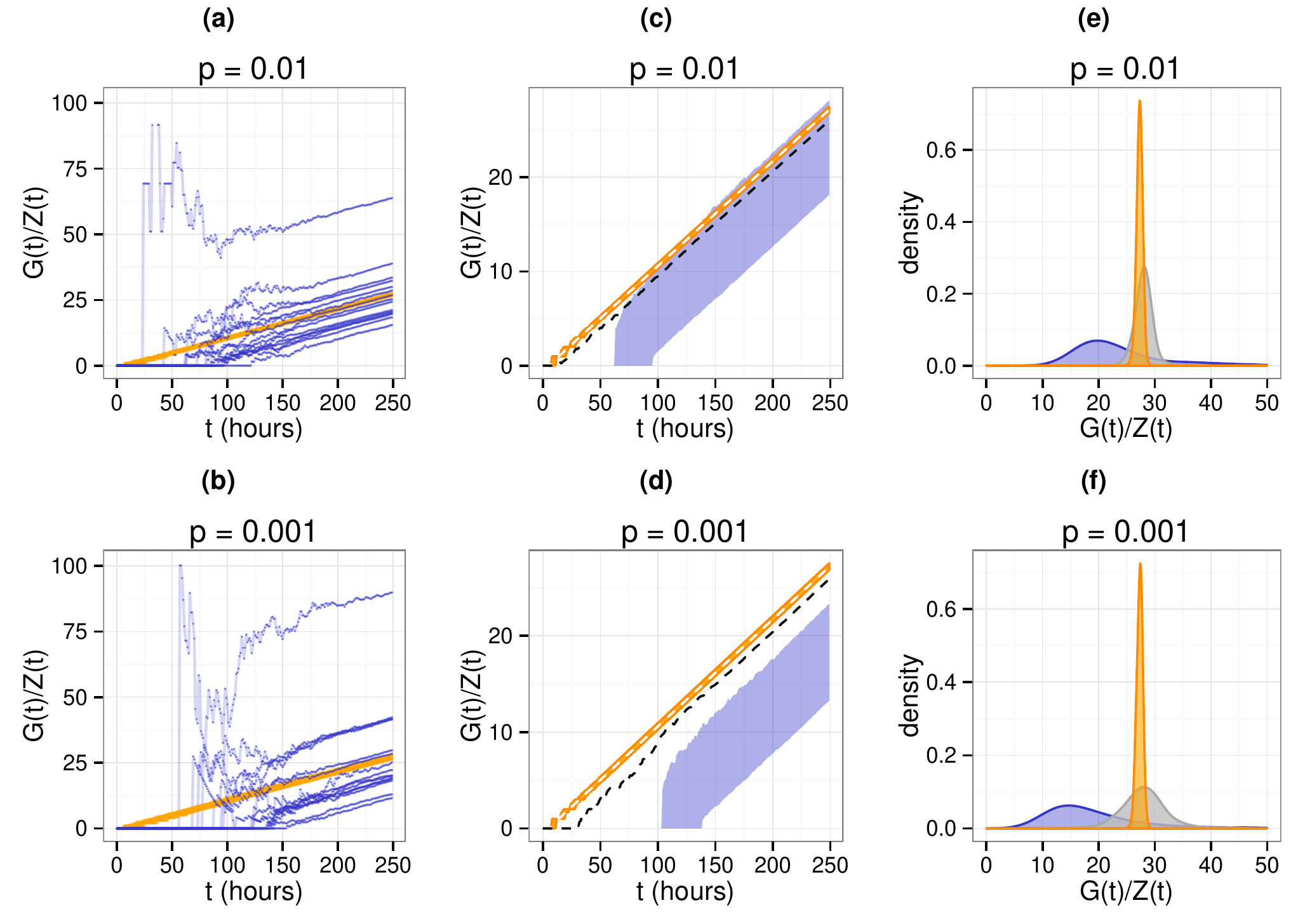}
\caption{
{\bf Single progenitor simulation validation.}
Simulations, described in detail in 
Section \ref{sec:val_sim},
start with a single label-positive cell at time $0$. Plots shown
for two label-change probabilities, $p=10^{-2}$ (top row) and
$p=10^{-3}$
(bottom row).
For $15$ independent realizations, orange lines display the actual
average generation, $G(t)/Z(t)$. Blue lines display the estimates
$-1/p\log(\Zp(t)/Z(t))$ from \eqref{eq:approx}, which ultimately
increase linearly with the same slope as the average generation,
but initially show significant variability.
(c)-(d)
For $100$ independent realizations of the process, the dashed line
reveals that the empirical average of the estimates is close to the
true value. The blue region is an inter-quartile plot of the
$-1/p\log(\Zp(t)/Z(t))$ estimates, with the upper boundary being
where 25 of the realizations are larger and the lower boundary being
where 25 of realizations are smaller, which typically exhibits
under-estimation.
(e)-(f)
Using Monte Carlo methods and a standard kernel estimator, these
panels show, for $t_2=250$ hours, the density of observations of
$G(t_2)/Z(t_2)$ in orange, the single time-point estimate
$-1/p\log(\Zp(t_2)/Z(t_2))$ in blue, and the two time-point estimate
\eqref{eq:two_tps} in grey with an additional time point at $t_1=200$
hours. The additional time-point improves the estimate by removing
early variability on a realization-by-realization basis.
}
\label{fig:single_progenitor}
\end{figure}

\begin{figure}
\centering\includegraphics[scale=0.7]{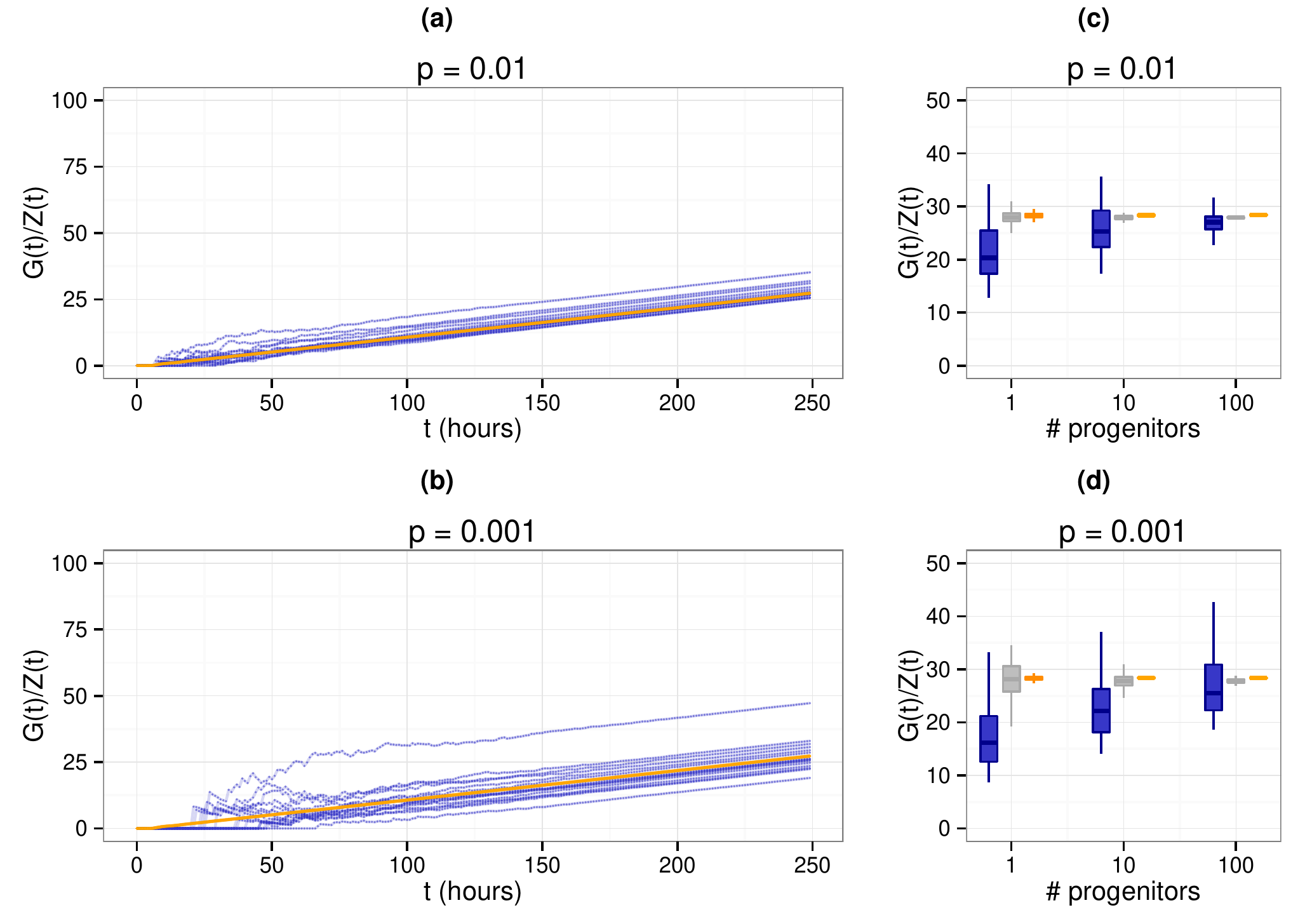}
\caption{
{\bf Multiple progenitor in silico validation.} Plots shown for two
label-change probabilities, $p=10^{-2}$ (top row) and $p=10^{-3}$
(bottom row).
(a)-(b) 
Sample paths of $15$ independent realizations with $100$ progenitors.
Orange lines correspond to average generation, $G(t)/Z(t)$. Blue
lines display the estimates $-1/p\log(\Zp(t)/Z(t))$, whose early
time fluctuation is significantly less pronounced than with
a single progenitor, as in \Fig{fig:single_progenitor} (a)-(b).
(c)-(d)
For 1, 10 and 100 progenitors, Monte Carlo generated box plots of
average generation (orange) and single time-point estimates (blue)
at $t=250$ hours, and two time-point estimates (grey), with the
additional measurement at $t_1=200$ hours. As predicted by theory,
accuracy improves substantially, in a $p$
dependent fashion, as the number of progenitors increases.
}
\label{fig:multiple_progenitors}
\end{figure}

\begin{figure*}
\centering\includegraphics[scale=0.5]{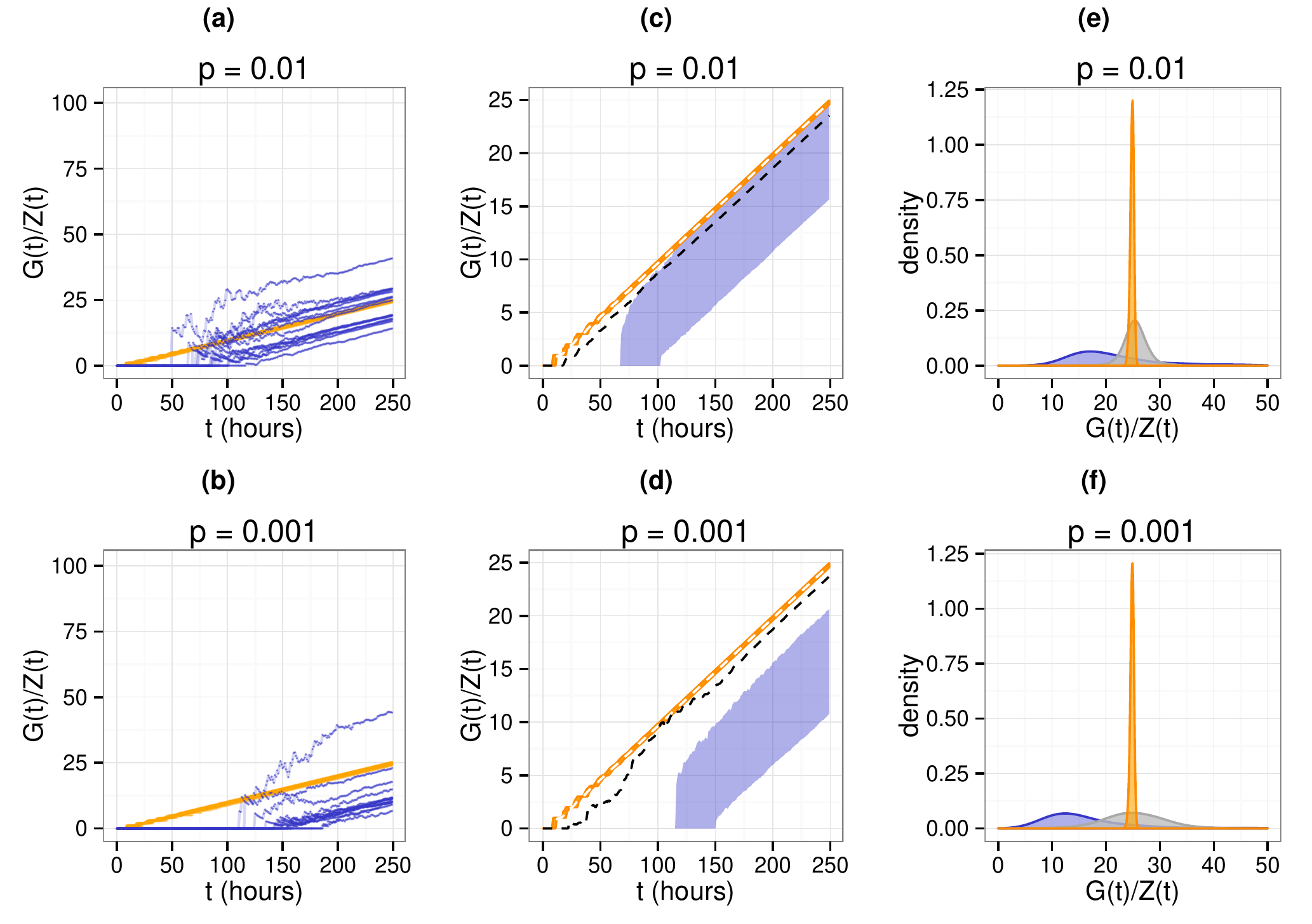}
\caption{Similar results to those shown in \Fig{fig:single_progenitor},
but using a cell life-time distributed according to a delayed
Gamma distribution with shape parameter of $3$, scale parameter of
$1$ and a delay of 7 hours.
}
\label{fig:single_progenitor_gam}
\end{figure*}
\begin{figure*}
\centering\includegraphics[scale=0.5]{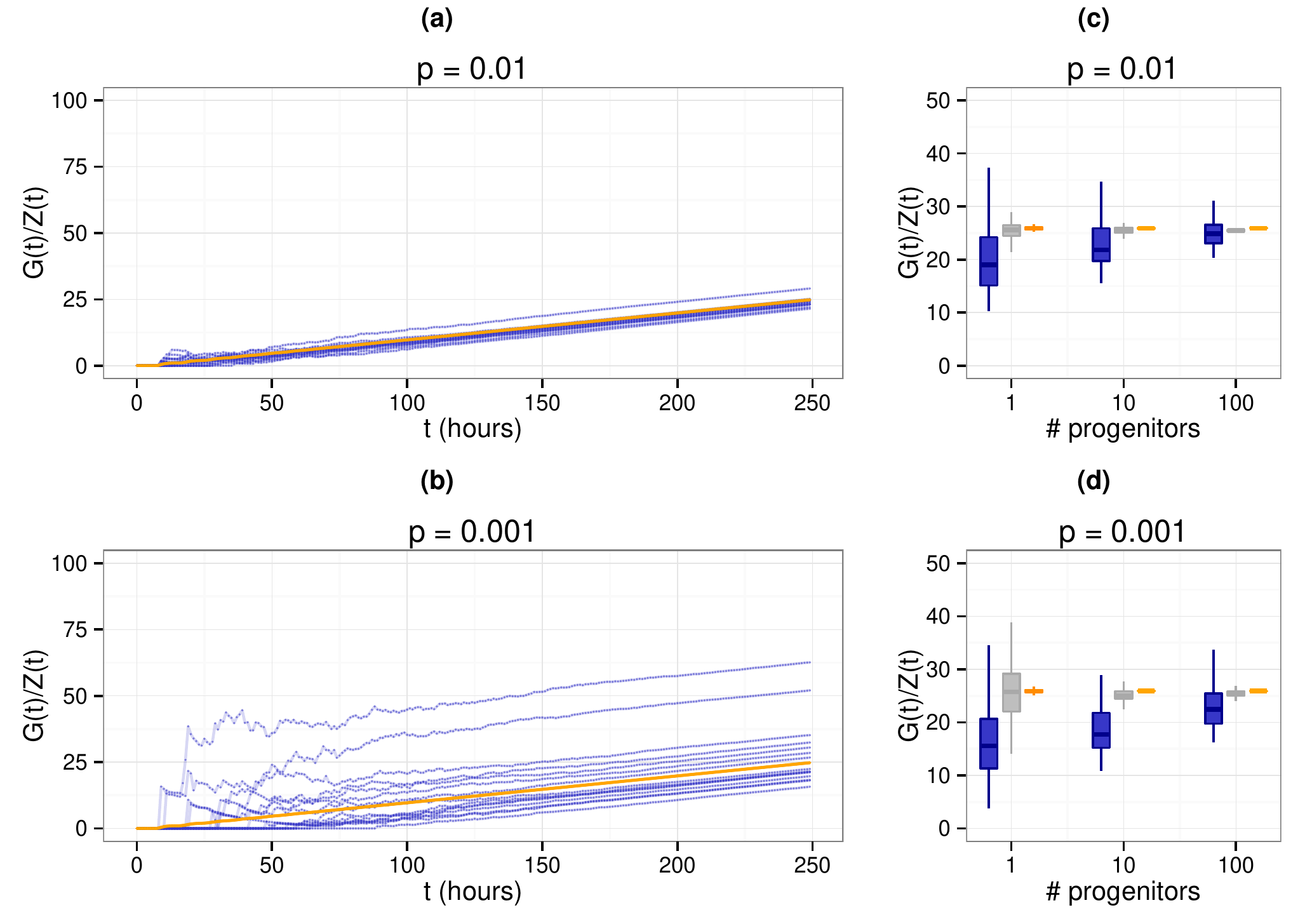}
\caption{Similar results to those shown in \Fig{fig:multiple_progenitors},
but using a cell life-time distributed according to a delayed
Gamma distribution with shape parameter of $3$, scale parameter of
$1$ and a delay of 7 hours.
}
\label{fig:mult_progenitor_gam}
\end{figure*}

\begin{figure}
\centering\includegraphics[scale=0.7]{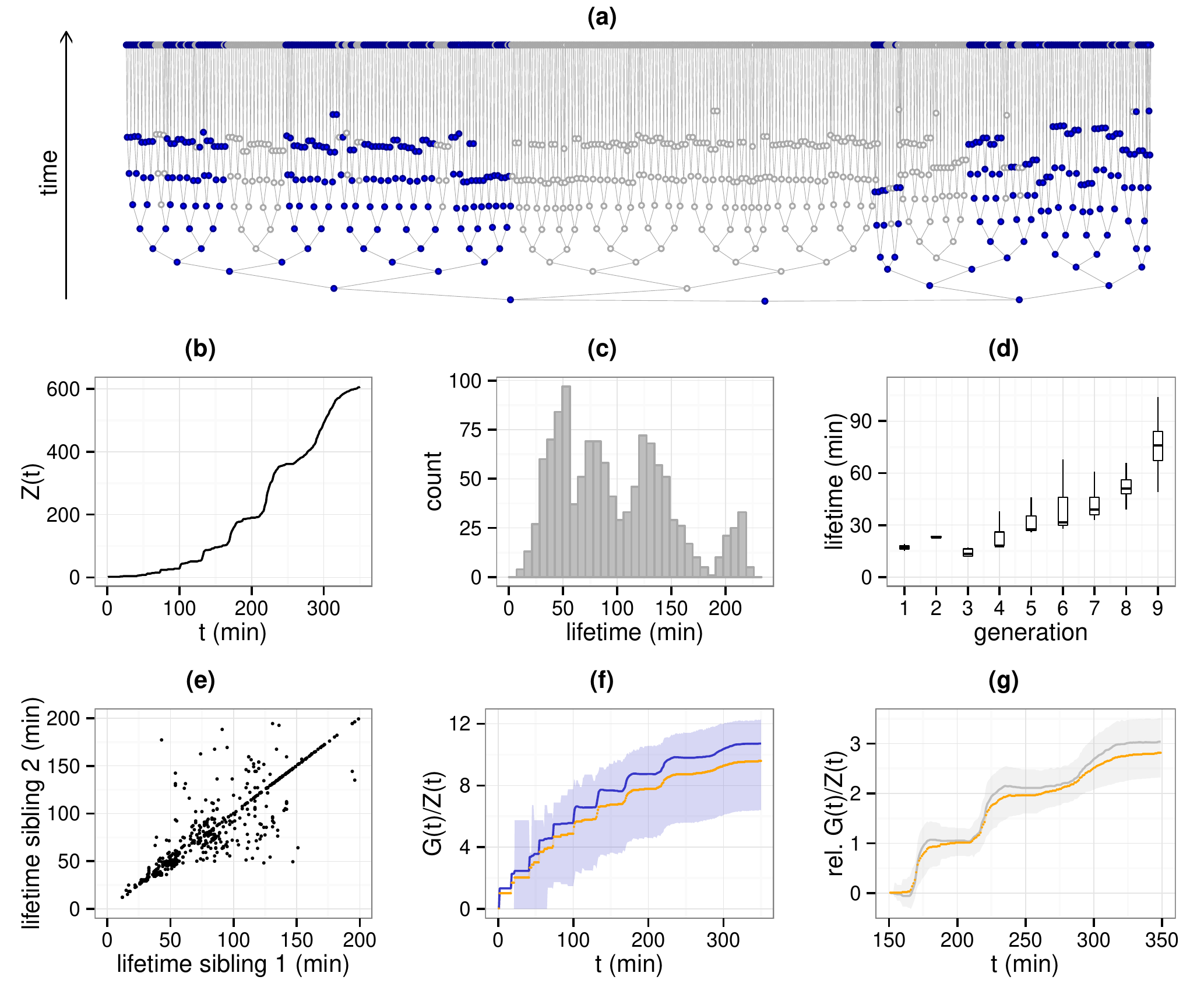}
\caption{
{\bf Validation with data from time-lapse microscopy}.
C. elegans embryonic development tree taken
from three-dimensional time-lapse movies of embryos recorded over
approximately six hours, supplementary material of \cite{richards2013},
and decorated with a stochastic labeling with $p=0.05$.
(a)
Lineage tree of C. elegans embryos, for first nine
generations. The lowest node represents the progenitor present in
the first frame, while the uppermost nodes represent cells alive
in the last frame. Each node in between represents a division event,
and its vertical position corresponds to the frame in which
it was recorded. The color of the nodes, blue for label-positive
at division and white for label-negative at division, illustrate
one possible random delabeling of the trees. 
(b) Total population, $Z(t)$, as a function of time, $t$, shows a
step-wise increase in cell numbers followed by a reduction in growth
rate that is not consistent with a homogeneous age dependent branching
process.
(c) Histogram of the cell lifetime shows a multi-modal distribution,
partly explained by the fact that lifetimes appear to lengthen with
time of birth or generation.
(d) Box plot of the lifetimes as a function of generation, 
which shows lifetimes increasing with generation.
(e) Sibling lifetimes are highly correlated, with many dividing
during the same frame.
(f) Comparison of average generation $G(t)/Z(t)$ (orange)
with its estimate $-\log\left(\Zp(t)/Z(t)\right)/p$ (blue).  Monte
Carlo determined inter-quartile ranges based on $10^4$ samples. 
(g) Same as (f), but starting at $t_1=150$ and computing the
difference of the average generation $G(t)/Z(t)$ (orange) 
since $t_1$ via the approximation \eqref{eq:two_tps}.  }
\label{fig:celegans}
\end{figure}

\begin{figure}
\centering\includegraphics[scale=0.6]{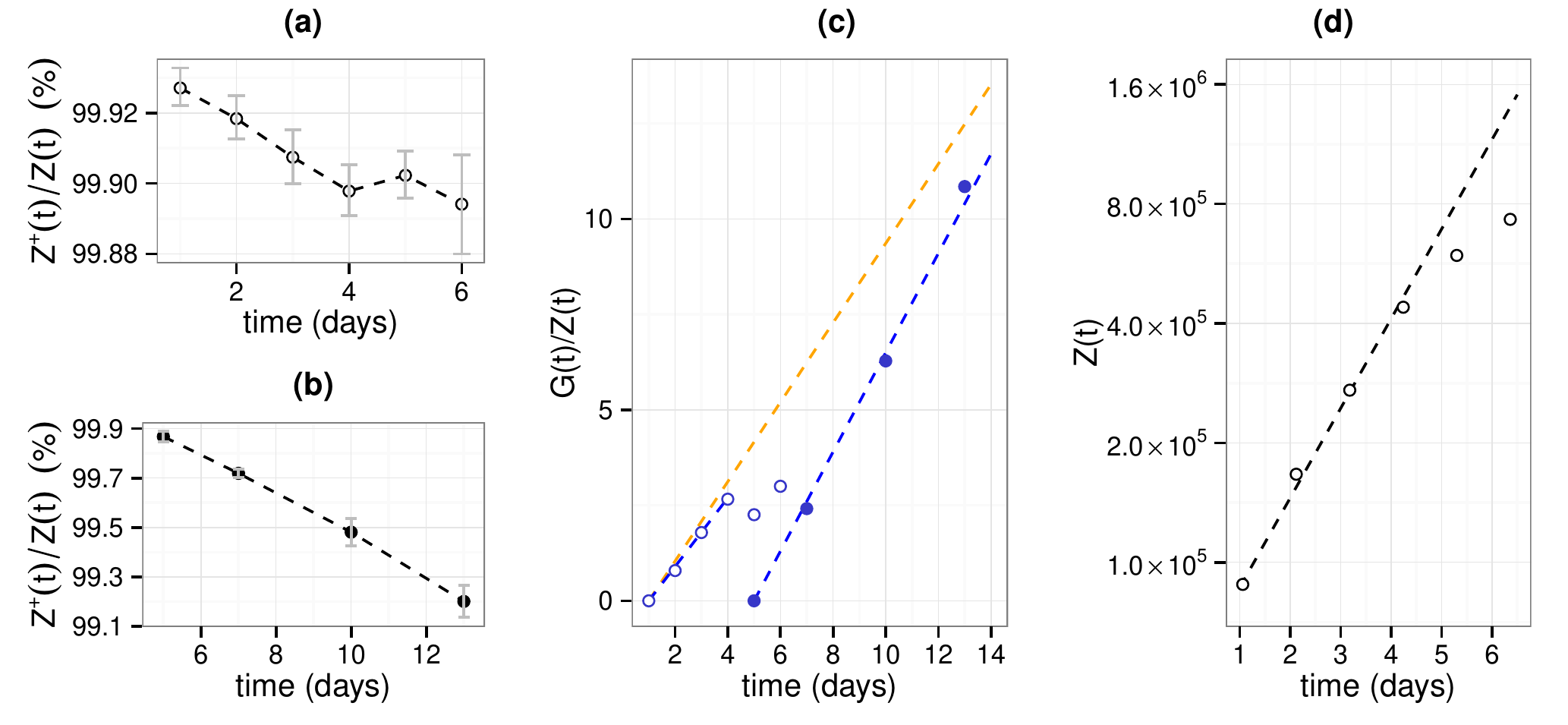}
\caption{
{\bf Validation with data from micro-satellite reporter systems.}
(a-b) Percentage of unlabeled cells $\Zp(t)/Z(t)$ over several
days of culture for genetically modified clones of MEF cells (a, \cite{Kozar2013})
and HCT116 cells (b, \cite{Gasche2003}) carrying micro-satellite mutation
reporter cassettes. The mutation probabilities, as reported in the
respective studies are $1.1\times 10^{-4}$ \cite{Kozar2013} and $6.1\times 10^{-4}$
\cite{Gasche2003}, which we use for $p$, the probability of delabeling
per cell division. Data are represented as mean $\pm$ SEM. For the
HCT116 data, the mean at each time point is over three clones. (c)
Inferred relative average generation G(t)/Z(t) from the proportions
of unlabeled cells and $p$ using the estimator $1/p\log(\Zp(t)/Z(t))$
(MEF, open circles; HCT116, filled circles). For both data sets the
average generation increases linearly with time, as expected from
the theory. In addition, the slope is close to one generation per
day, even thought the mutation probability $p$ differs by a factor
of six. The orange dashed line shows the increase in average generation
as estimated directly from the growth curve $Z(t)$ (panel d) assuming exponentially
distributed life-times and no cell death (for details see main text),
demonstrating the accuracy of the estimator. The blue dashed regression
lines ($r^2=0.99$ for both sets) highlight the linear increase of average generation estimates during the exponential growth phase. (d) Population size $Z(t)$ of cultured MEF cells (open circles, \cite{Kozar2013}) and exponential growth fit to first four time points (dashed line, $r^2=0.97$).
\label{fig:MS}} \end{figure} 

\begin{figure}
\centering\includegraphics[scale=0.7]{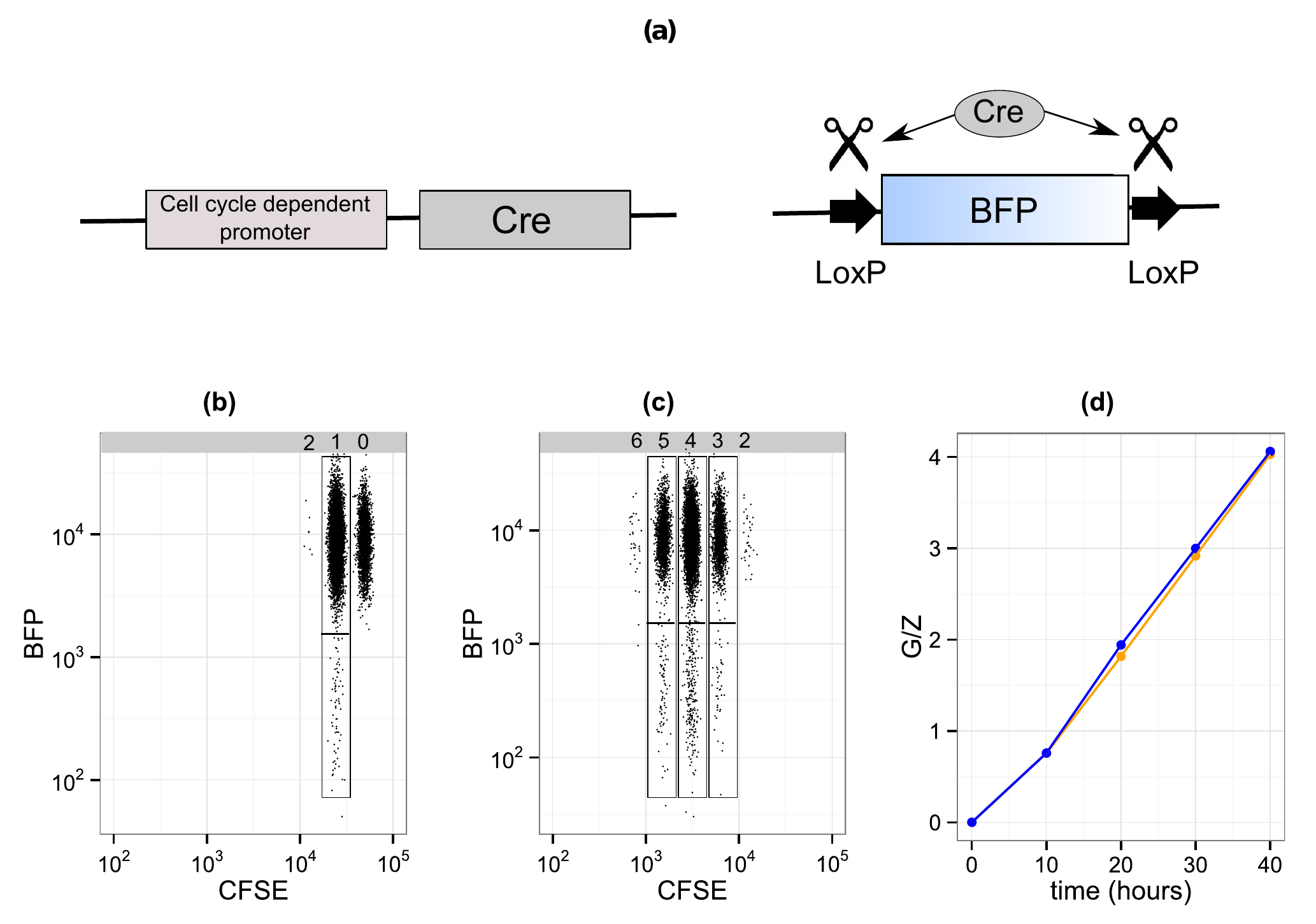}
\caption{
{\bf Proposed probabilistic division-linked label-loss construct.}
(a) 
Schematic drawing of the construct. Cells initially express a
protein such as Blue Fluorescent Protein (BFP). During each cell
cycle Cre is expressed and with a small probability the gene for
BFP is excised, leading to irreversible label-loss. (b)-(d) In
silico simulation of proposed FACS experiment to determine the
probability of label-loss, $p$. The simulation is parameterized as
in Section \ref{sec:val_sim}. $20,000$ BFP+ cells are labeled
with a division diluting dye such as CFSE, which is green, and
cultured until they divide. At each division, CFSE is shared evenly
between daughters and the BFP label is lost with probability $p=0.01$.
(b)-(c) 
Simulated FACS plots, ten and 40 hours after CFSE labeling.
The fluorescence signals were set to realistic experimental values
\cite{hawkins07} and noise was added to initial CFSE levels for
illustrative purposes.  The gates in (b) and (c) capture the
proportion of label-positive cells, $\BFPp(n)$, in each generation,
$n$. From these, the label-loss probability, $p$, can be experimentally
determined via the formula $p=1-\BFPp(n+1)/\BFPp(n)$. As all cells
in the zeroth generation express BFP, $\BFPp(0) = 1$, and so
$p=1-\BFPp(1)=\BFPn(1)=0.0101$, the proportion of BFP negative cells in
generation $1$.
(d) 
The orange line plots the average generation of the population,
determined from the simulated CFSE data, a function of time. The
blue line is estimate determined using the method described in the
present paper \eqref{eq:approx} with a label-loss parameter,
$p=0.0101$, obtained from (b). The CFSE determined value
and the estimate exhibit excellent concordance.
}
\label{fig:construct}
\end{figure}

\end{document}